\documentclass[12pt]{article}
\RequirePackage{amsthm,amsmath,amsbsy,amsfonts}
\usepackage{bm,graphicx,psfrag,epsf}
\usepackage{enumerate}
\usepackage{natbib}
\usepackage{paralist, booktabs, multirow}

\newtheorem{lemma}{Lemma}
\newcommand{\blind}{0}

\numberwithin{equation}{section}
\theoremstyle{plain}

\usepackage{amssymb}
\usepackage[usenames]{color}

\usepackage[mathscr]{eucal}

\usepackage{stmaryrd}

\usepackage{xspace}

\usepackage{paralist}

\usepackage{ifpdf}

\usepackage{graphicx}
\ifpdf
\usepackage{epstopdf}
\DeclareGraphicsExtensions{.eps,.pdf}
\else
\usepackage{epsfig}
\fi

\usepackage{subfig}

\ifpdf
\usepackage[colorlinks,urlcolor=blue,citecolor=blue,linkcolor=blue]{hyperref}
\else
\newcommand{\href}[2]{{#2}}
\usepackage{url}
\fi

\usepackage{tabulary}

\usepackage{algorithm}
\usepackage{algpseudocode}

\ifpdf
\newcommand{\Sec}[1]{\hyperref[sec:#1]{Section~\ref*{sec:#1}}} 
\newcommand{\App}[1]{\hyperref[sec:#1]{Appendix~\ref*{sec:#1}}} 
\newcommand{\Eqn}[1]{\hyperref[eq:#1]{{\rm (\ref*{eq:#1})}}} 
\newcommand{\Part}[1]{\hyperref[part:#1]{(\ref*{part:#1})}} 
\newcommand{\Fig}[1]{\hyperref[fig:#1]{Figure~\ref*{fig:#1}}} 
\newcommand{\Tab}[1]{\hyperref[tab:#1]{Table~\ref*{tab:#1}}} 
\newcommand{\Thm}[1]{\hyperref[thm:#1]{Theorem~\ref*{thm:#1}}} 
\newcommand{\Lem}[1]{\hyperref[lem:#1]{Lemma~\ref*{lem:#1}}} 
\newcommand{\Prop}[1]{\hyperref[prop:#1]{Proposition~\ref*{prop:#1}}} 
\newcommand{\Cor}[1]{\hyperref[cor:#1]{Corollary~\ref*{cor:#1}}} 
\newcommand{\Def}[1]{\hyperref[def:#1]{Definition~\ref*{def:#1}}} 
\newcommand{\Alg}[1]{\hyperref[alg:#1]{Algorithm~\ref*{alg:#1}}} 
\newcommand{\Ex}[1]{\hyperref[ex:#1]{Example~\ref*{ex:#1}}} 
\newcommand{\As}[1]{\hyperref[as:#1]{Assumption~{\rm\ref*{as:#1}}}} 
\newcommand{\Reg}[1]{\hyperref[as:#1]{Condition~\ref*{reg:#1}}} 
\newcommand{\AlgLine}[2]{\hyperref[alg:#1]{line~\ref*{line:#2} of Algorithm~\ref*{alg:#1}}}
\newcommand{\AlgLines}[3]{\hyperref[alg:#1]{lines~\ref*{line:#2}--\ref*{line:#3} of Algorithm~\ref*{alg:#1}}}
\else
\newcommand{\Sec}[1]{{Section~\ref{sec:#1}}} 
\newcommand{\App}[1]{{Appendix~\ref{sec:#1}}} 
\newcommand{\Eqn}[1]{{(\ref{eq:#1})}} 
\newcommand{\Part}[1]{{(\ref{part:#1})}} 
\newcommand{\Fig}[1]{{Figure~\ref{fig:#1}}} 
\newcommand{\Tab}[1]{{Table~\ref{tab:#1}}} 
\newcommand{\Thm}[1]{{Theorem~\ref{thm:#1}}} 
\newcommand{\Lem}[1]{{Lemma~\ref{lem:#1}}} 
\newcommand{\Prop}[1]{{Property~\ref{prop:#1}}} 
\newcommand{\Cor}[1]{{Corollary~\ref{cor:#1}}} 
\newcommand{\Def}[1]{{Definition~\ref{def:#1}}} 
\newcommand{\Alg}[1]{{Algorithm~\ref{alg:#1}}} 
\newcommand{\Ex}[1]{{Example~\ref{ex:#1}}} 
\newcommand{\As}[1]{{Assumption~\ref{as:#1}}} 
\newcommand{\Reg}[1]{{R~\ref{reg:#1}}} 
\newcommand{\AlgLine}[2]{{line~\ref{line:#2} of Algorithm~\ref{alg:#1}}}
\newcommand{\AlgLines}[3]{{lines~\ref{line:#2}--\ref{line:#3} of Algorithm~\ref{alg:#1}}}
\fi

\usepackage{paralist, booktabs, multirow}

\newtheorem{assumption}{Assumption}[section]

\newtheorem{theorem}{Theorem}[section]

\usepackage[normalem]{ulem}

\usepackage{mathtools}

\usepackage{algorithm}
\usepackage{algpseudocode}
\usepackage{tikz}
\usetikzlibrary{arrows}
\newtheorem{proposition}{Proposition}[section]

\newcommand{\Real}{\mathbb{R}}

\newcommand{\Tra}{^{\sf T}} 
\newcommand{\Inv}{^{-1}} 
\newcommand{\tr}{\operatorname{tr}} 
\newcommand{\Ker}{\operatorname{Ker}} 
\def\vec{\mathop{\rm vec}\nolimits}
\newcommand{\amp}{\mathop{\:\:\,}\nolimits}
\newcommand{\bic}{\text{BIC}}

\newcommand{\V}[1]{{\bm{\mathbf{\MakeLowercase{#1}}}}} 
\newcommand{\VE}[2]{\MakeLowercase{#1}_{#2}} 
\newcommand{\Vhat}[1]{{\bm{\hat \mathbf{\MakeLowercase{#1}}}}} 
\newcommand{\Vtilde}[1]{{\bm{\tilde \mathbf{\MakeLowercase{#1}}}}} 
\newcommand{\Vn}[2]{\V{#1}^{(#2)}} 
\newcommand{\VnE}[3]{{#1}^{(#2)}_{#3}} 

\newcommand{\M}[1]{{\bm{\mathbf{\MakeUppercase{#1}}}}} 
\newcommand{\ME}[2]{\MakeLowercase{#1}_{#2}} 

\newcommand{\Mn}[2]{\M{#1}^{(#2)}} 
\newcommand{\MnE}[3]{\MakeLowercase{#1}^{(#2)}_{#3}} 

\newcommand{\Kron}{\otimes} 

\begin{document}

\def\spacingset#1{\renewcommand{\baselinestretch}%
{#1}\small\normalsize} \spacingset{1}


\if0\blind
{
  \title{\bf Going off the Grid:  Iterative Model Selection for Biclustered Matrix Completion}
  \author{Eric C. Chi\thanks{
    Department of Statistics, North Carolina State University, Raleigh, NC 27695-8203 (E-mail: eric$\_$chi@ncsu.edu)} \,
   Liuyi Hu\thanks{
    Department of Statistics, North Carolina State University, Raleigh, NC 27695-8203 (E-mail: lhu@ncsu.edu)} \,    
    Arvind K. Saibaba\thanks{Department of Mathematics, North Carolina State University, Raleigh, NC 27695-8203 (E-mail: asaibab@ncsu.edu)} \,
    and
    Arvind U. K. Rao\thanks{Department of Bioinformatics and Computational Biology, Division of Quantitative Sciences, The University of Texas MD Anderson Cancer Center, Houston, TX,  77030 (E-mail: aruppore@mdanderson.org).}    \\}
    \date{}
  \maketitle
} \fi

\if1\blind
{
  \bigskip
  \bigskip
  \bigskip
  \begin{center}
    {\LARGE\bf Title}
\end{center}
  \medskip
} \fi

\bigskip
\begin{abstract}
We consider the problem of performing matrix completion with side information on row-by-row and column-by-column similarities. We build upon recent proposals for matrix estimation with smoothness constraints with respect to row and column graphs. We present a novel iterative procedure for directly minimizing an information criterion in order to select an appropriate amount row and column smoothing, namely perform model selection. We also discuss how to exploit the special structure of the problem to scale up the estimation and model selection procedure via the Hutchinson estimator. We present simulation results and an application to predicting associations in imaging-genomics studies.
\end{abstract}

\noindent%
\vfill

\newpage
\spacingset{1.45} 

\section{Introduction}
\label{sec:introduction}

In the matrix completion problem, we seek to recover or estimate a matrix, when only a fraction of its entries are observed. While it is impossible to complete an arbitrary matrix using only partial observations of its entries, it may be possible to fully recover matrix entries when the matrix has an appropriate underlying structure. For example, most low-rank matrices can be completed accurately with high probability, by solving a convex optimization problem~\citep{CanRec2009}. Consequently, algorithms for low-rank matrix completion have enjoyed widespread use across many disciplines, including collaborative filtering and recommender systems \citep{Koren2009}, multi-task learning and classification \citep{Amit2007, Argyriou2007, Wu2015}, computer vision \citep{Chen2004}, statistical genetics \citep{Chi2013}, as well as remote sensing \citep{Malek-Mohammadi2014}.

In this paper, we consider matrix completion under a structural assumption that is closely related to the low-rank assumption; i.e.\@, we assume that the matrix entries vary ``smoothly'' with respect to a graphical organization of the rows and columns. For example, in the context of a movie recommendation system, we seek to complete a user-by-movies ratings matrix. We may have additional information about users, such as if pairs of users are friends on a social media application, as well as additional information from a movie database, such as the co-occurrence of certain film principles. We expect the entries of a movie ratings matrix to vary ``smoothly'' over a neighborhood of users, defined by a friendship graph, and over a neighborhood of movies, defined by a shared movie principles graph. When such local similarity structure exists, and is available, it behooves us to leverage this information to predict missing entries in a matrix.

In general, we wish to recover a matrix $\M{Z} \in \Real^{n \times p}$ from a noisy and partially observed matrix $\M{X} \in \Real^{n \times p}$ when there exist similarities between pairs of rows and pairs of columns. Let the parameters $w_{ij} = w_{ji}$ and $\tilde{w}_{ij} = \tilde{w}_{ji}$ for $i=1,\dots,n$ and $j=1,\dots,p$ denote non-negative weights that quantify the similarities between pairs of rows and pairs of columns. 
Let $\Omega \subset \{1, \ldots, n\} \times \{1, \ldots, p\}$ denote the set of observed indices. Finally, let  $\mathcal{P}_\Omega(\M{Z})$ denote the projection operator onto the set of indices $\Omega$ where the $ij$th entry of $\mathcal{P}_\Omega(\M{Z})$ is $\ME{z}{ij}$ if $(i,j) \in \Omega$ and is zero otherwise. With this notation in hand, we can pose this version of the matrix completion task as the following optimization problem:
\begin{equation}
\label{eq:bmc}
\underset{\M{Z} \in \Real^{n\times p}}{\min}\; \ell(\M{Z}) + J(\M{Z}),
\end{equation}
where
\begin{eqnarray*}
\ell(\M{Z}) & \equiv & \frac{1}{2} \lVert \mathcal{P}_\Omega(\M{X}) - \mathcal{P}_\Omega(\M{Z}) \rVert_{\text{F}}^2 \quad \text{and}\\
J(\M{Z}) & \equiv & \frac{\gamma_r}{2}
 \sum_{i<j}w_{ij} \lVert \M{Z}_{i \cdot}-\M{Z}_{j \cdot} \rVert^2_2 + \frac{\gamma_c}{2} \sum_{i<j}\tilde{w}_{ij} \lVert \M{Z}_{\cdot i}-\M{Z}_{\cdot j} \rVert_2^2.
\end{eqnarray*}



In the equations above, $\M{Z}_{\cdot i}$ ($\M{Z}_{\cdot i})$ denotes the $i$th row (column) of the matrix $\M{Z}$ and $(\gamma_r, \gamma_c)$ are nonnegative regularization parameters. The first term $\ell(\M{Z})$ quantifies the misfit between $\M{Z}$ and $\M{X}$ over the observed entries $\Omega$. 
The second term $J(\M{Z})$ is a penalty that incentivizes smoothness with respect to the row and column similarities. The two nonnegative parameters $(\gamma_r, \gamma_c)$ control the relative importance of minimizing the discrepancy between $\M{Z}$ and $\M{X}$ over $\Omega$, and enforcing smoothness of $\M{Z}$ with respect to the given row and column similarities. We refer to the matrix completion problem given in \Eqn{bmc} as the biclustered matrix completion (BMC) problem. Several variations on \Eqn{bmc} have been proposed in the literature prior to this work \citep{Ma2011,Cai2011,Kalofolias2014, Rao2015,ShaPerKal2016}, and smoothness penalties similar to $J(\M{Z})$ have been applied in penalized regression \citep{LiLi2008, RanNovLan2014,Hu2015,Li2016} and functional principal components analysis \citep{HuaSheBuj2009, TiaHuaShe2012, AllGroTay2014}.

\subsection{Contributions}
Our major contributions in this paper are two-fold. First, we derive some new properties of BMC problem, concerning the existence and uniqueness of a solution and as well as the solution's limiting behavior as the penalty parameters tend to infinity. Second, we provide a computational framework for model selection, namely choosing $(\gamma_r,\gamma_c)$. We survey the contents of this paper, emphasizing the main results.

\paragraph{Properties of BMC} Despite the widespread use of the graph smoothing penalties like $J(\M{Z})$ in matrix completion, we present new results on basic properties of the regularizer $J(\M{Z})$, the BMC optimization problem \Eqn{bmc}, and the BMC solution. Many of these results, while intuitive, have been taken for granted without careful justification. A key consequence of these results is that they highlight when BMC also recovers a low-rank matrix. This fact suggests that BMC may be more computationally advantageous over other variants proposed in the literature. Additionally, these results also suggest strategies to sparsify the row and column weights in order to speed up estimation while still ensuring that the BMC problem is well defined. Specifically, we show that the BMC problem always has a solution and give conditions on the missingness pattern and row and column weights that guarantee the solution's uniqueness. Furthermore, we show that as the regularization parameters diverge to infinity, the BMC solution converges to a limiting smooth estimate of the data matrix and also derive what this limit is.

\paragraph{Computational framework for model selection}

The optimization problem in \Eqn{bmc} is convex and differentiable. The solution for a fixed set of regularization parameters $(\gamma_r, \gamma_c)$ requires solving a linear system. As we will see later, this system admits a unique solution under conditions that can be easily verified. We study, in detail, the problem of solving this linear system for a fixed set of parameters, as well as choosing optimal parameters $(\gamma_r, \gamma_c)$, i.e.\@, perform model selection. The prevalent approach to choosing these parameters is by searching for a minimizer of a surrogate measure of predictive performance over a two-way grid of candidate parameters. Common surrogate measures include prediction error on hold-out sets, as in cross-validation, and various information criteria. Cross-validation in particular is popular since it is easy to implement \citep{Rao2015,ShaPerKal2016}. While grid-search may be computationally feasible for choosing a single parameter, it can be prohibitively expensive when selecting two parameters since each grid point requires fitting a model for those parameters, and 
in the case of BMC fitting a model requires inverting a potentially very large matrix. Moreover, it requires pre-specifying a grid of regularization parameters.

\begin{figure}[ht]
    \centering
	\begin{tabular}{cccc}
		\subfloat[Iterative Model Selection (IMS) Path]{\includegraphics[scale=0.38]{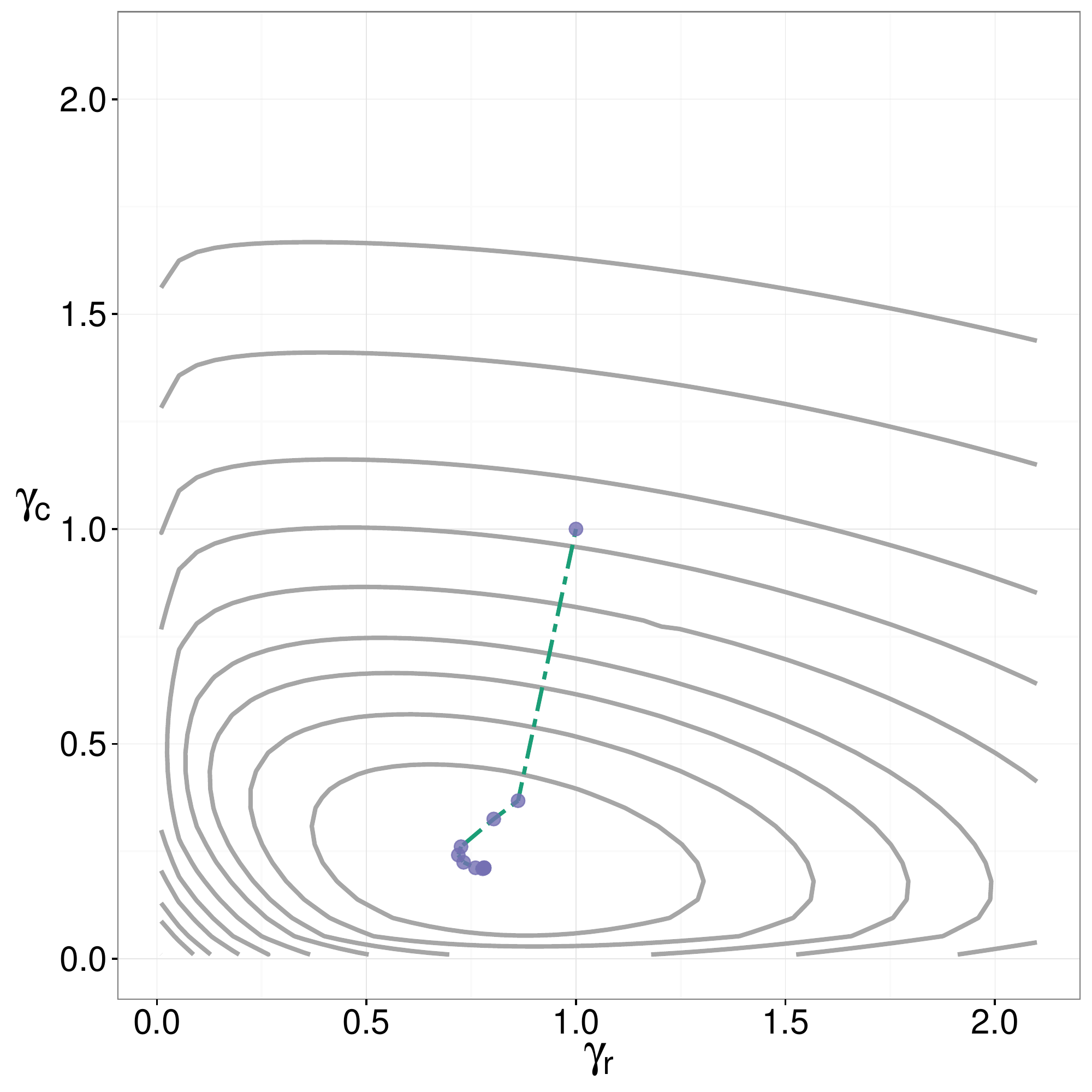}
		\label{fig:search_path}} 
	    & \subfloat[Grid-Search]{\includegraphics[scale=0.38]{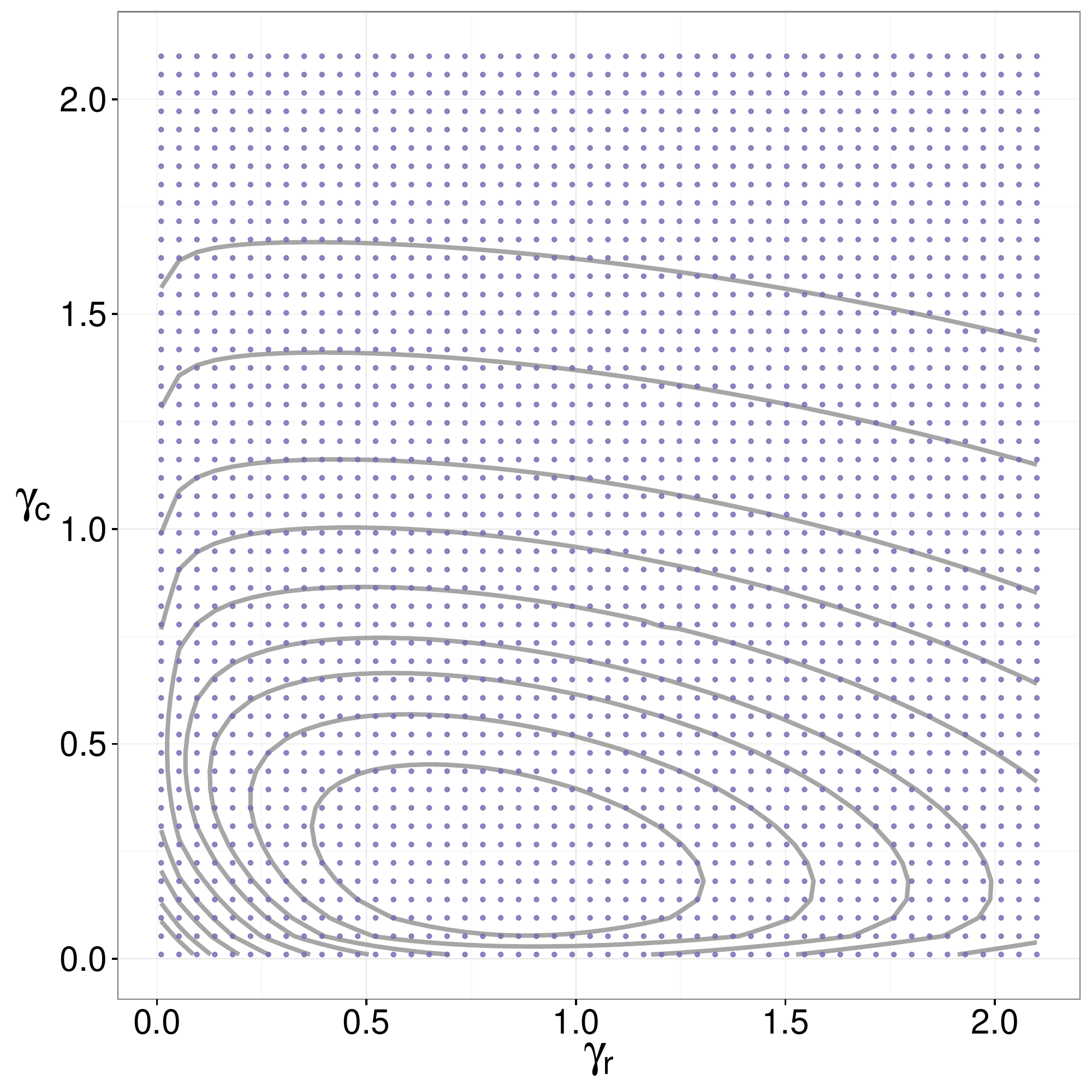}
		\label{fig:bic_grid}} \\			
	\end{tabular}
	\caption{Searching for the minimizer of the BIC in order to find regularization parameters for completing a 100-by-100 matrix. IMS requires 12 iterations to converge to the minimizer; each iteration's most expensive step requires solving a $10^4$-by-$10^4$ linear system. Searching the 50-by-50 grid requires solving $2,500$ different $10^4$-by-$10^4$ linear systems. \label{fig:search_bic}}
\end{figure}

Our second contribution is a novel scalable strategy for model selection in BMC problems based on directly minimizing the Bayesian Information Criterion (BIC). The BIC for \Eqn{bmc} is continuously differentiable and is amenable to minimization by Quasi-Newton methods. To further scale up our procedure, we introduce a refinement based on the Hutchinson estimator to approximate the BIC, and then minimize this approximation. Our resulting procedure, which we call Iterative Model Selection (IMS), leads to drastic reduction in the computational time to select $(\gamma_r, \gamma_c)$ and does not require pre-specifying a grid of regularization parameter pairs to explore. \Fig{search_path} shows an example of the search path taken by IMS exploring the BIC surface on one of the simulated problems described in \Sec{simulated_data}. IMS took 12 iterations to converge to the minimum. Consider searching the BIC surface over a 50-by-50 grid of candidate parameters. \Fig{bic_grid} shows the set of 50-by-50 grid points at which the BIC would have to be evaluated. Each evaluation requires solving a large linear system. As we show later in \Sec{model_selection}, the dominant calculation at each IMS iteration is solving the same linear system. While similar smoothing parameters would be chosen by the two procedures, this simple example illustrates how the naive grid-search may blindly evaluate the BIC at many points far from a minimum and therefore may unnecessarily solve far more linear systems than the IMS. In this example, grid-search would solve 2,500 linear systems, while IMS would solve 12 to arrive at essentially the same model.

To summarize, the IMS path sports the follows advantages over the standard grid-search: (i) In practice, it often takes a more direct route to a model minimizing the BIC leading to potentially many fewer linear system solves, (ii) it does not not require pre-specifying the grid, and (iii) consequently, model selection is not restricted to a finite set of pre-specified grid points. In short, by enabling the model search to go off the tuning parameter grid, we can perform similar and sometimes superior model selection while also reaping significant savings in computation time.  


The rest of the paper is organized as follows. In \Sec{prior_art}, we review the relationship of the BMC problem to the prior art in matrix completion. In \Sec{bmc_solution}, we present new results on properties of the BMC solution.
In \Sec{estimation}, we discuss the problem of solving \Eqn{bmc} for a fixed pair of regularization parameters. In \Sec{model_selection}, we frame the model selection problem and discuss how to efficiently search the regularization parameter space with IMS to select a model with good prediction accuracy. In \Sec{monte_carlo}, we elaborate on how to further scale up IMS using stochastic approximation strategies. In \Sec{Examples}, we present an empirical comparison of IMS and standard grid-based regularization parameter selection methods on both simulated as well as a real data example from radiogenomics. In \Sec{conclusion}, we close with a discussion.

\section{Relationship to Prior Art}
\label{sec:prior_art}

To put BMC into context and clarify its connections to prior art, we review the two primary formulations of matrix completion in the literature: low-rank matrix completion (LMRC) and matrix completion on graphs (MCG).

\paragraph{Low-Rank Matrix Completion (LRMC)} In the noisy LRMC problem, we seek to recover a denoised matrix $\M{Z} \in \Real^{n \times p}$ from a noisy and incomplete matrix $\M{X} \in \Real^{n \times p}$ by solving the following constrained optimization problem:
\begin{equation}
\label{eq:lrmc}
\text{minimize}\; \ell(\M{Z}) \quad\quad \text{subject to} \quad \quad \text{rank}(\M{Z}) \leq r.
\end{equation}

This formulation balances the tradeoff between how well $\M{Z}$ matches $\M{X}$ over the observed entries $\Omega$ and model complexity of $\M{Z}$ as measured by its rank. As we relax the bound on the rank $r$ by making it larger, we can better fit the data at risk of overfitting it.

Due to the rank constraint, \Eqn{lrmc} is a combinatorial optimization problem and quickly becomes impractical to solve as the problem size increases. Fortunately, we can solve the following computationally tractable convex problem instead:
\begin{eqnarray}
\label{eq:nuc}
\text{minimize}\; & \frac{1}{2} \lVert \mathcal{P}_\Omega(\M{X}) - \mathcal{P}_\Omega(\M{Z}) \rVert_{\text{F}}^2 + \gamma_n \lVert \M{Z} \rVert_*
\end{eqnarray}
As before in \Eqn{bmc}, the first term quantifies how well $\M{Z}$ approximates $\M{X}$ over the observed entries $\Omega$. The second term $\lVert \M{Z} \rVert_*$ denotes the nuclear norm of $\M{Z}$, which is the sum of its singular values, and the nonnegative regularization parameter $\gamma_n$ trades off the emphasis on these two terms. Problem \Eqn{nuc} is related to problem \Eqn{lrmc} through the fact that the nuclear norm of a matrix is the tightest convex approximation to its rank \citep{Fazel2002}. Remarkably, under suitable conditions on the missingness patterns defined by $\Omega$, the solution to the convex problem in \Eqn{nuc} also coincides with those of the combinatorial problem in \Eqn{lrmc} with high probability \citep{Candes2010}. 

\paragraph{Matrix Completion on Graphs (MCG)} 
Given how successful the low-rank paradigm is, a natural strategy for incorporating information on row and column similarities would be to augment \Eqn{nuc} with the penalty $J(\M{Z})$ and solve the following convex optimization problem:
\begin{eqnarray}
\label{eq:mcg}
\underset{\M{Z} \in \Real^{n\times p}}{\min}\; \ell(\M{Z}) + \frac{\gamma_n}{2} \lVert \M{Z} \rVert_* + J(\M{Z}).
\end{eqnarray}

With respect to BMC, the only difference between \Eqn{bmc} and \Eqn{mcg} is the addition of a nuclear norm penalty in \Eqn{mcg}. While the problem defined in \Eqn{mcg} is also convex, including the nuclear norm penalty drastically complicates the estimation procedure. Solving LRMC is tractable because there exist polynomial time iterative solvers. Nonetheless, iterative solvers for \Eqn{nuc} and consequently \Eqn{mcg} typically require computing an expensive singular value decomposition (SVD) to account for the nuclear norm regularizer.  Considerable attention has been given to either formulate alternative non-convex optimization problems that omit the nuclear norm penalty entirely~\citep{Burer2003,Srebro2005, Rao2015}, or performing judiciously chosen {\em low-rank} SVD calculations \citep{MazHasTib2010, CaiCanShe2010}. Moreover, there are now three tuning parameters $(\gamma_n,\gamma_r, \gamma_c)$ that tradeoff the emphasis on the data-fit and the structure imposed on $\M{Z}$. Given the costs of including the nuclear norm penalty, a natural question to ask is how much added benefit is gained by including it?  

The following illustrative example provides some evidence that the penalty $J(\M{Z})$ is typically sufficient for completion tasks when the matrices exhibit strong row and column clustering structure. Such matrices exhibit a checkerboard or biclustered structure under row and column reordering. 

\begin{figure}[t]
    \centering
	\begin{tabular}{cccc}
		\subfloat[Original]{\includegraphics[scale=0.11]{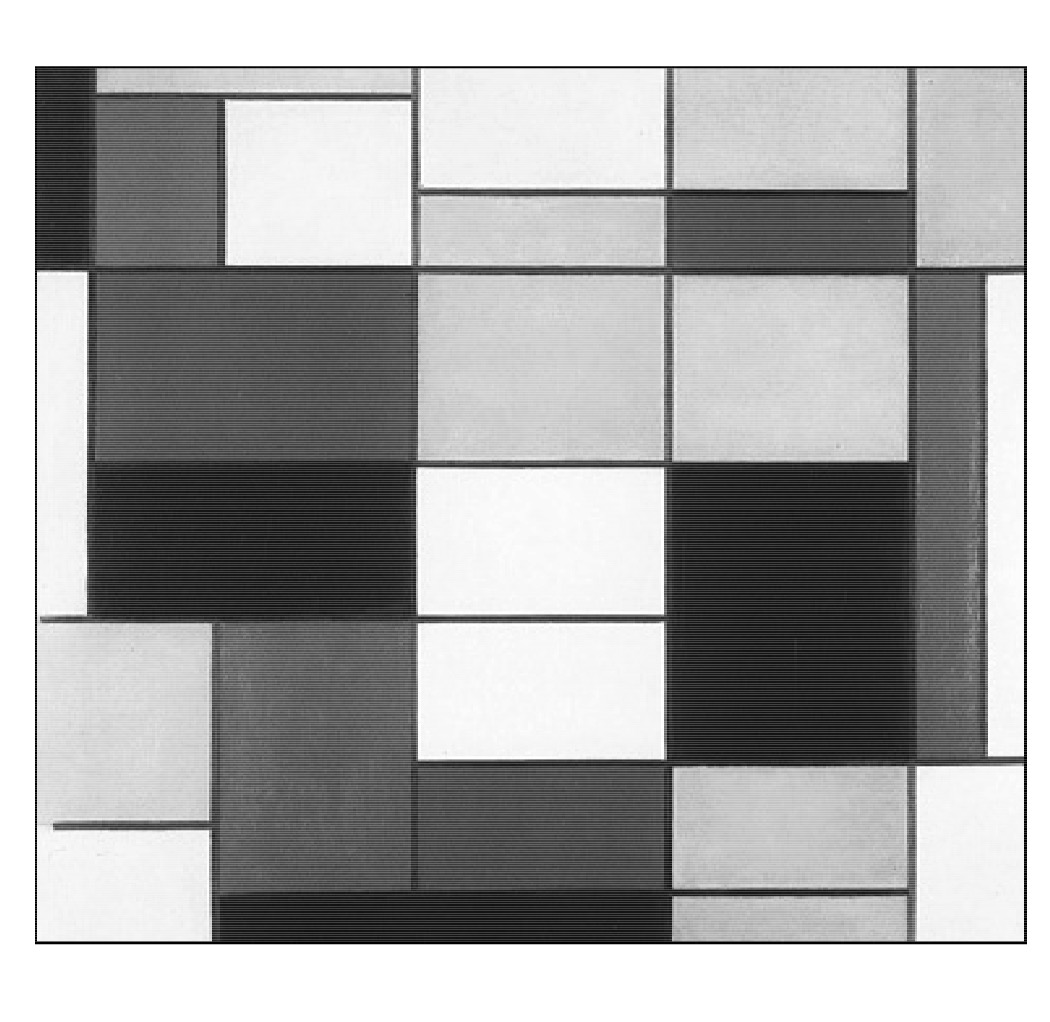}
		\label{fig:mondrian_true}} 
	    & \subfloat[Noise + Missing 50\%]{\includegraphics[scale=0.11]{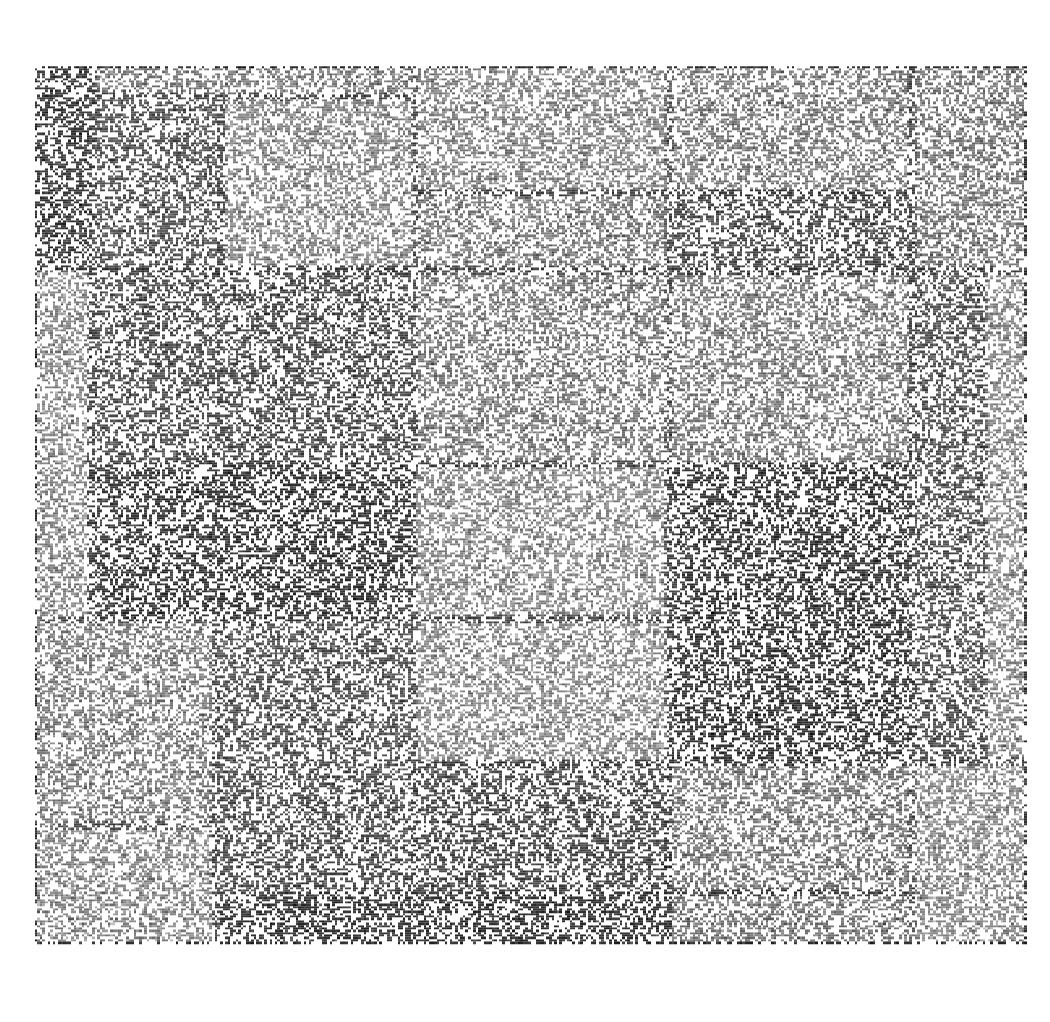}
		\label{fig:mondrian_corrupted}}
	  &  \subfloat[Low-Rank Completion]{\includegraphics[scale=0.11]{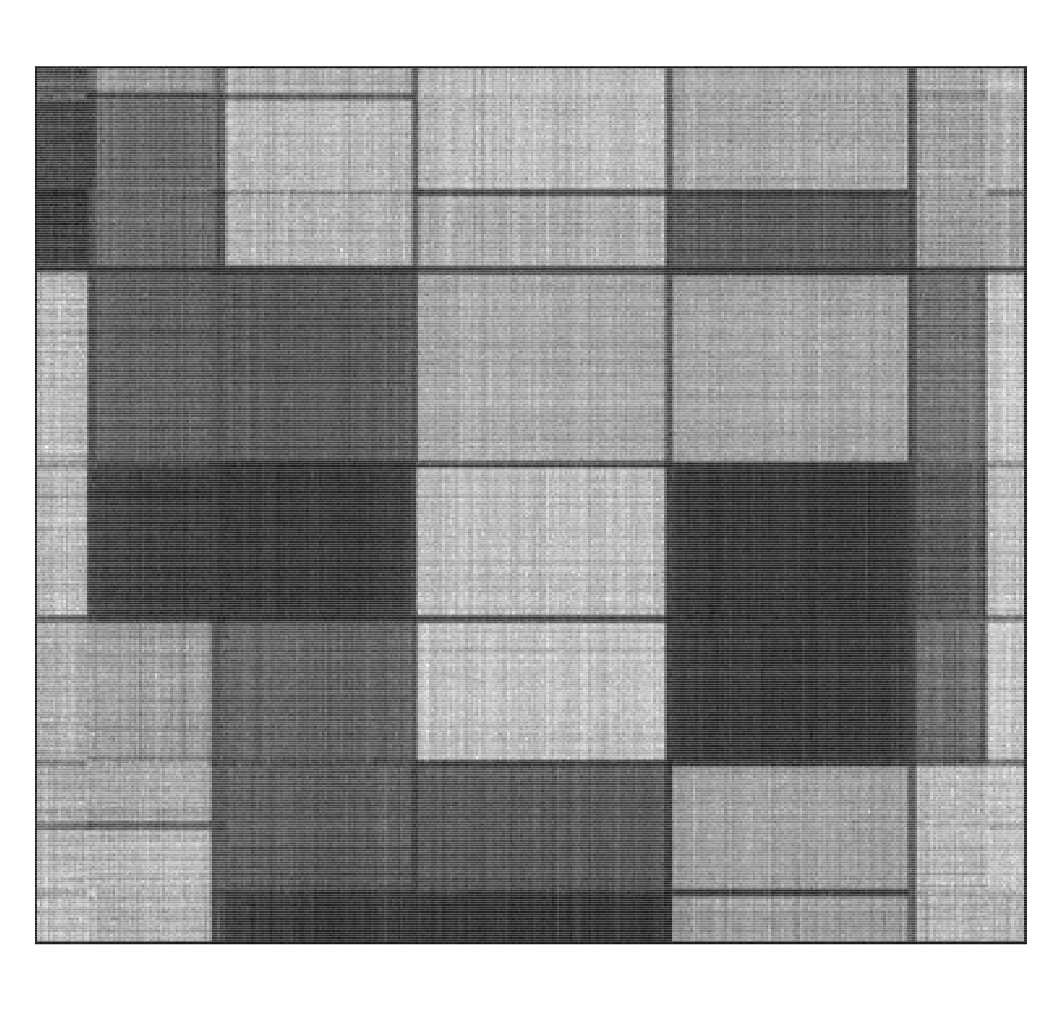}
		\label{fig:mondrian_svt}}
	    &\subfloat[Biclustered Completion]{\includegraphics[scale=0.11]{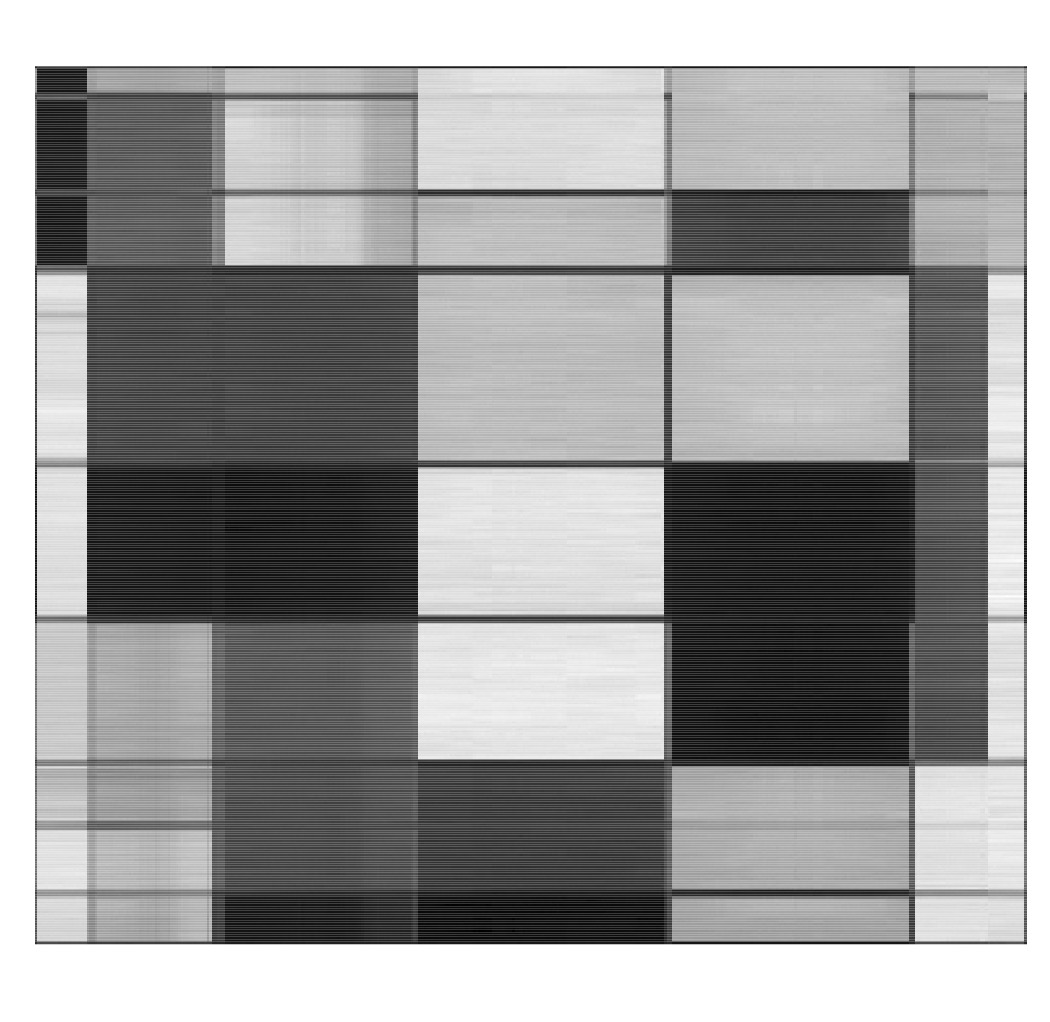}
		\label{fig:mondrian_bmc}}\\ 				
	\end{tabular}
	\caption{Composition A by Piet Mondrian. The matrix is 370-by-380. Each element takes on an integer value between 0 and 255. We added i.i.d.\ $\mathcal{N}(0,\sigma^2)$ noise where $\sigma = 50$ and removed 50\% of the entries. Missing entries were then estimated using low-rank matrix completion and biclustered matrix completion \label{fig:mondrian}}
\end{figure}
\Fig{mondrian} compares the results from performing LRMC and BMC on a digital replica of the oil painting `Composition A' (\Fig{mondrian_true}) by the Dutch painter Piet Mondrian\footnote{A jpeg file was obtained from \href{http://www.wikiart.org/en/piet-mondrian/composition-a-1923}{http://www.wikiart.org/en/piet-mondrian/composition-a-1923}.} after adding noise and removing half its entries (\Fig{mondrian_corrupted}). To the eye, both LRMC (\Fig{mondrian_svt}) and BMC (\Fig{mondrian_bmc}) appear to give reasonably good reconstructions. Further inspection reveals that the BMC predictions have lower mean squared error over the unobserved entries than LRMC. Details on this experiment and the MSE calculations are in the Supplementary Materials.

The comparable performance of LRMC and BMC on this example suggests that the nuclear norm penalty in the MCG problem may be an unnecessary computational complication when there is an underlying biclustering structure. Indeed, we will see next that the penalty $J(\M{Z})$ shrinks solutions towards a low-rank matrix defined by the connectivity structure of the underlying row and column graphs. 

\section{Properties of the BMC Solution}
\label{sec:bmc_solution}

The BMC formulation is related to recent work by \cite{ShaPerKal2016}; however, they present results from a signal processing perspective. In contrast, our perspective is on shrinkage estimation. Furthermore, the results on matrix completion are new. All proofs are in \Sec{proof} of the Appendix.

To better understand the action of $J(\M{Z})$, we need to review some basic facts from algebraic graph theory. Let $\mathcal{G} = (V, E)$ denote an undirected graph with a vertex set $V = \{1, \ldots, n\}$ and an edge set $E = V \times V$. A weighted undirected graph also includes a non-negative weight
function $w : V \times V \rightarrow \Real_+$ that is symmetric in its arguments, namely $w(i,j) = w(j,i)$. 
The set $A \subset V$ is a connected component of $\mathcal{G}$ if (i) there is a sequence of edges forming a path between every pair of vertices in $A$ and (ii) none of its vertices are connected to any vertices in its complement $V \backslash A$. Let $\V{\chi}_A$ denote the indicator function on the set of vertices $A \subset \mathcal{V}$, namely $\V{\chi}_A(i) = 1$ if $i \in A$ and $\V{\chi}_A(i) = 0$ if $i \not \in A$. 
Recall that the graph Laplacian $\M{L} \in \Real^{n \times n}$ of $\mathcal{G}$ is a symmetric positive semidefinite matrix given by
\begin{eqnarray*}
\ME{L}{ij} & = & \begin{cases}
\underset{(i,i') \in E}{\sum} \ME{w}{ii'} & \text{if $i = j$} \\
- w_{ij} & \text{otherwise.}
\end{cases}
\end{eqnarray*}
Define a weighted undirected row graph $\mathcal{G}_r = (V_r, E_r)$ with $V = \{1, \ldots, n\}$ and weights $w_{ij}$, and denote its graph Laplacian by $\M{L}_r \in \Real^{n \times n}$. We use analogous notation for a weighted undirected column graph $\mathcal{G}_r$.

It is straightforward to show that the regularizer $J(\M{Z})$ can be expressed in terms of the two graph Laplacians, as
\begin{eqnarray*}
J(\M{Z}) & = & \frac{\gamma_r}{2}\tr(\M{Z}\Tra\M{L}_r\M{Z}) + \frac{\gamma_c}{2}\tr(\M{Z}\M{L}_c\M{Z}\Tra).
\end{eqnarray*}
The expression above explicitly characterizes the shrinkage action of $J(\M{Z})$ in terms of the connectivity properties of $\mathcal{G}_r$ and $\mathcal{G}_c$. We present a result which gives conditions under which the penalty $J(\M{Z}) =0$.
\begin{proposition}
\label{prop:shrinkage}
Suppose that there are $R$ row connected components $A_1, \ldots, A_R$ in $\mathcal{G}_r$ and $C$ column connected components $B_1, \ldots, B_C$ in $\mathcal{G}_c$
Then the penalty $J(\M{Z}) = 0$ if and only if $\M{Z} = 
\sum_{r=1}^R \sum_{c=1}^C \mu_{rc} \V{\chi}_{A_r} \V{\chi}_{B_c}\Tra$ for some $\mu_{rc}$ for $r = 1, \ldots, R$ and $c = 1, \ldots, C$.
\end{proposition}

\Prop{shrinkage} suggests that the penalty $J(\M{Z})$ incentivizes approximations of $\M{X}$ whose rows and columns are spanned by the indicator functions of the connected components of the row and column graphs $\mathcal{G}_r$ and $\mathcal{G}_c$. In other words, $J(\M{Z})$ shrinks estimates to matrices that are piecewise constant on submatrices defined by the functions $\V{\chi}_{A_r}\V{\chi}_{B_r}\Tra$. We refer to these submatrices as biclusters or checkerboard patches. Indeed, 
suppose that the data matrix is a linear combination of the outer products of the indicator functions of the row and column connected components, namely
\begin{eqnarray}
\label{eq:bicluster_model}
\M{X} & = & \sum_{r=1}^R \sum_{c=1}^C \mu_{rc} \V{\chi}_{A_r} \V{\chi}_{B_c}\Tra.
\end{eqnarray}
Given \Prop{shrinkage}, we intuitively expect that the BMC estimate should be able to exactly recover missing entries in this scenario. This is indeed the case, provided the missingness pattern is reasonable. We make explicit what we mean by reasonable in the following assumption, which will be invoked throughout the rest of this paper.

\begin{assumption}
  \label{as:missingness}
  If there are $R$ row connected components $A_1, \ldots, A_R$ in $\mathcal{G}_r$ and $C$ column connected components $B_1, \ldots, B_C$ in $\mathcal{G}_c$, then $\mathcal{P}_\Omega (\V{\chi}_{A_r} \V{\chi}_{B_c}\Tra) \not = \V{0}$ for all $r = 1, \ldots, R$ and $c = 1, \ldots, C$.
\end{assumption}
In words, \As{missingness} states that every checkerboard patch defined by a pair of row and column connected components must have at least one observation. Under this assumption and the ideal scenario presented in \Eqn{bicluster_model}, the BMC estimate of the missing entries is exact.
\begin{proposition}
\label{prop:exact_recovery}
Suppose that \As{missingness} holds. Then $\M{Z} = \M{X}$ in \Eqn{bicluster_model} is the unique global minimizer to \Eqn{bmc} for all positive $\gamma_r$ and  $\gamma_c$.
\end{proposition}
There are two important observations about the form of $\M{X}$ in \Eqn{bicluster_model}.
First, $\M{X}$ that can be expressed as in \Eqn{bicluster_model} corresponds to the checkerboard pattern we seek to recover. Second, such $\M{X}$ are low-rank, when the number of row clusters $R < n$ and the number of column clusters $C < p$ and consequently $\M{X}$ in \Eqn{bicluster_model} has rank at most $RC \ll np$. The second observation motivates employing the simpler BMC over MCG when the underlying matrix has a biclustered structure. 

The penalty $J(\M{Z})$ is already shrinking solutions towards a low-rank solution, likely rendering the addition of a nuclear norm penalty a computationally expensive redundancy.  
Of course, this is an ideal case when the data matrix $\M{X}$ has the form in \Eqn{bicluster_model}. We bring it up mainly to understand (i) what $J(\M{Z})$ is shrinking estimates towards, (ii) when the nuclear norm  may be unnecessary, and consequently (iii) for what kind of data matrices BMC is best equipped to recover. These results suggest that BMC should perform well when the true underlying matrix has an approximately checkerboard pattern and row and column weights that are consistent with that pattern can be supplied. Experiments in \Sec{Examples} will confirm this suspicion. For now though, we turn our attention to the properties of the BMC problem and solution for a general data matrix $\M{X}$ and general set of row and column weights.


Our first main result concerns the existence and uniqueness of the solution to the BMC problem \Eqn{bmc}.

\begin{theorem}
\label{thm:existence}
A solution to the BMC problem \Eqn{bmc} always exists. The solution is unique if and only if \As{missingness} holds and $\gamma_r$ and $\gamma_c$ are strictly positive. If \As{missingness} does not hold, then there are infinitely many solutions to \Eqn{bmc}.
\end{theorem}

The interpretation of this result is that there is a unique solution to the biclustered matrix completion problem if and only if no bicluster induced by the row and column graph Laplacians is completely missing. On the other hand, the prediction error for the reconstruction can be arbitrarily poor if \As{missingness} fails to hold.

In order for \Thm{existence} to be practical, however, we need a way to verify \As{missingness}. We provide an algorithm based on breadth-first-search that accomplishes this in time linear in the size of the data. Details are given in the \Sec{verify} of the Appendix. The next two results characterize the limiting behavior of $\M{Z}(\gamma_r, \gamma_c)$ as a function of the tuning parameters $(\gamma_r, \gamma_c)$.


%

Since $J(\M{Z})$ is shrinking estimates towards the checkerboard pattern induced by the clustering pattern in the row and column graphs, we intuitively expect that the estimate $\M{Z}(\gamma_r, \gamma_c)$ tends toward the solution of the following constrained optimization problem:
\begin{eqnarray}
\label{eq:constrained}
\M{Z}^\star = \underset{\M{z}}{\arg\min}\; \frac{1}{2} \lVert \mathcal{P}_\Omega (\M{z}) - \mathcal{P}_\Omega (\M{x}) \rVert_{\text{F}}^2,\quad \text{subject to}\quad \tr (\M{Z}\Tra\M{L}_r\M{Z}) = \tr (\M{Z} \M{L}_c \M{Z}\Tra) = 0.
\end{eqnarray}
Moreover, we anticipate that this limiting solution should be the result of averaging the observed entries over each checkerboard patch. This is indeed the case.
\begin{proposition}
\label{prop:constrained_solution}
If \As{missingness} holds, then the unique solution to \Eqn{constrained} is 
\begin{equation}
\label{eq:solution}
\M{Z}^\star \amp = \amp \sum_{r=1}^R\sum_{c=1}^C \mu_{rc}^*\V{\chi}_{A_r}\V{\chi}_{B_c}\Tra,
\end{equation}
where $\Omega_{rc} = \{ (i,j) \in \Omega : i \in A_r, j \in B_c \}$, and \begin{equation}
\label{eq:global}
\mu_{rc}^* \amp \equiv \amp \lvert \Omega_{rc} \rvert\Inv \sum_{(i,j) \in \Omega_{rc}} \ME{x}{ij}.
\end{equation}
\end{proposition}

The next result verifies our intuition that the estimate $\M{Z}(\gamma_r, \gamma_c)$ tends towards to $\M{Z}^\star$ in \Eqn{solution} as $\gamma_r$ and $\gamma_c$ tend towards infinity.

\begin{theorem}
\label{thm:limiting_solution}
If \As{missingness} holds, then $\M{z}(\gamma_r, \gamma_c)$ tends to $\M{Z}^\star$, defined in \Eqn{solution}, as the regularization parameters diverge to infinity, namely $\underline{\gamma} \equiv \min\{\gamma_r, \gamma_c\} \rightarrow \infty$.
\end{theorem}

\section{Estimation}
\label{sec:estimation}

In this section, we discuss how to solve the estimation problem for a fixed set of parameters $(\gamma_r, \gamma_c)$ in order to quantify the amount of work a standard grid-search method would incur.
 It is easier to work with vectorized quantities. Let $\V{x} \equiv \vec(\M{X})$, namely $\V{x}$ is the vector obtained by stacking the columns of $\M{X}$ on top of each other.  Then the objective in \Eqn{bmc} can be written as
\begin{eqnarray}
\label{eq:smooth_bicluster_vec}
&\frac{1}{2} \lVert \M{P}_\Omega \V{X} - \M{P}_\Omega \V{Z} \rVert_{2}^2 + \frac{\gamma_r}{2} \V{z}\Tra (\M{I}\Kron\M{L}_r)\V{z} + \frac{\gamma_c}{2} \V{Z}\Tra(\M{L}_c\Kron\M{I})\V{z},
\end{eqnarray}
where the binary operator $\Kron$ denotes the Kronecker product and $\M{P}_\Omega \in \{0,1\}^{np \times np}$ is a diagonal matrix with a 1 in the $k$th diagonal entry if the $k$th entry in the matrix (column major ordering) is observed and 0 otherwise. More explicitly, if $(i,j) \in \Omega$, then $\M{P}_\Omega(k,k) = 1$ where $k = i + n(j-1)$. We have rewritten the two penalty expressions in terms of $\V{z}$ by invoking the identity $\vec(\M{A}\M{B}\M{C}) = (\M{C}\Tra \Kron \M{A})\V{b}$.

Since the objective function in \Eqn{smooth_bicluster_vec} is differentiable and convex, we seek the vector $\V{z}$ at which the gradient of the objective vanishes. Thus, the estimate $\V{Z}$ is the solution to the following linear system obtained by setting the gradient equal to zero,
\begin{eqnarray*}
\label{eq:smooth_lin}
\M{S}\V{z} & = & \M{P}_\Omega \V{x},
\end{eqnarray*}
where $\M{S} = \M{P}_\Omega + \gamma_r (\M{I} \Kron \M{L}_r) + \gamma_c (\M{L}_c \Kron \M{I})$. Note that under \As{missingness}, this linear system is invertible and therefore $\V{z} = \M{S}\Inv \M{P}_\Omega\V{x}$. 

\paragraph{Sparse Weights}

Recall that $\M{S}$ is $np$-by-$np$. If the majority of the row and column weights are positive, the resulting Laplacian matrices will be dense and solving this $np$-by-$np$ linear system will take a demanding $\mathcal{O}( (np)^3)$ operations. On the other hand, if most of the row and column weights are zero, namely the weights are sparse, then we can solve the linear system in substantially less time as discussed below. Fortunately, it is possible to construct sparse approximations to the weights graphs that lead to Laplacian regularized solutions that are close to the solutions one would obtain using the original dense Laplacian regularizers \citep{Sadhanala2016}. Unless stated otherwise, for the rest of the paper we will assume that the weights are sparse. In particular, we assume the number of positive weights is linear in the size of the graph. This can be achieved using a $k$-nearest neighbors (knn) sparsification strategy described in \Sec{radiogenomics}.

Finally, we emphasize that there is a tension between minimizing computational costs and ensuring reliable estimation. If the weights are too sparse, \As{missingness} can fail to hold and there may not be a unique solution to the BMC problem \Eqn{bmc}. In practice, the knn sparsification for $k\sim10$ strikes a reasonable balance between these two goals. Again, we emphasize that we can easily check \As{missingness} to expedite the identification of a good sparsity level.

\paragraph{Computational Complexity} Since $\M{S}$ is symmetric and sparse, the solution to $\M{S}\V{z} = \M{P}_\Omega\V{x}$ can be computed using either a direct solver, such as the sparse Cholesky factorization of $\M{S}$, or an iterative solver, such as the preconditioned conjugate gradient.

In the direct approach, a triangular factorization of $\M{S}$ is computed and then forward and backward substitution are performed on two triangular systems to obtain the solution. The exact computational complexity of a sparse direct solver depends on the underlying sparsity pattern; however, a (knn) sparsity pattern in the matrix $\M{S}$ is similar to the discretization of elliptic partial differential equations in two spatial dimensions. For these problems, the computational complexity for solving a linear system involving $\M{S}$, which is of size $np$-by-$np$, requires $\mathcal{O}( (np)^{3/2})$ flops~\citep{george1973nested}. 

On the other hand, if the problem at hand is very large, factorization methods may not be feasible. We recommend that the linear system $\M{S}\V{z} = \M{P}_\Omega \V{x}$ be solved using Preconditioned Conjugate Gradient method, with the Incomplete Cholesky Factorization method as a preconditioner~\cite[Section 10.3]{golub2012matrix}. In numerical experiments, we have successfully used this approach for large-scale problems, which is directly available via MATLAB. 

\section{Model Selection}
\label{sec:model_selection}

We now address the issue of choosing the penalty parameters. We seek the parameters $(\gamma_r, \gamma_c)$ that result in the model with the best prediction error. There are two general approaches to estimating this prediction error: covariance penalties that are analytic and sampling methods that are non-parametric. The former includes methods such as Mallow's $C_p$, Akaike's information criterion (AIC), the Bayesian information criterion (BIC) \citep{Schwarz1978}, generalized cross-validation (GCV) \citep{Craven1978,Golub1979}, and Stein's unbiased risk estimate (SURE). The latter includes methods such as cross-validation (CV) and the bootstrap. Interested readers may consult \cite{Efron2004} for an in depth discussion on the relationship between the two approaches. In this article, we will use the BIC and compare it to CV since each are widely used in practice.
All proofs are in \Sec{proof} of the Appendix.

%

To derive the BIC for BMC, we first need to compute the degrees of freedom for BMC.  In general, the degrees of freedom of a model quantifies its flexibility. To derive the degrees of freedom for a model, we need a probabilistic model for the data generating process.
Suppose we observe noisy measurements of a parameter $\V{\mu}$, namely $\V{x} \in \Real^m$ where $\V{x} = \V{\mu} + \V{\varepsilon}$ and $\varepsilon_i$ are uncorrelated random variables with zero mean and common variance $\sigma^2$. Let $\Vhat{x}$ denote an estimate of $\V{\mu}$. Then the degrees of freedom of $\Vhat{x}$ is given by
\begin{eqnarray}
\label{eq:degrees_of_freedom}
\text{df} & = & \frac{1}{\sigma^2} \sum_{i=1}^m \text{Cov}(\hat{x}_{i},\VE{x}{i}).
\end{eqnarray}
In the case of BMC, this is a straightforward calculation. We assume the vectorization of our data, $\V{X} \in \Real^{np}$, is given by $\V{x} = \V{\mu} + \V{\varepsilon}$, where the elements of $\V{\varepsilon}$ are uncorrelated errors with zero mean and common variance $\sigma^2$. In the BMC problem, the vector $\V{z} = \M{S}\Inv\M{P}_\Omega\V{x}$ plays the role of $\Vhat{x}$ in the formula \Eqn{degrees_of_freedom}. Note that the degrees of freedom calculation does not require that the mean vector $\V{\mu}$ to have a checkerboard structure.

\begin{proposition}
\label{prop:dof}
The degrees of freedom of the BMC estimate is given by $\tr(\M{S}\Inv)$,
where
\begin{eqnarray*}
 \M{S} & \equiv & \M{P}_\Omega  + \gamma_r (\M{I} \Kron \M{L}_r) + \gamma_c (\M{L}_c \Kron \M{I}).
\end{eqnarray*}
\end{proposition}
The degrees of freedom for BMC has several intuitive properties. As $\underline{\gamma}$ diverges, the degrees of freedom decreases monotonically to the degrees of freedom of $\M{Z}^\star$, defined in \Eqn{solution}, namely $RC$ where $R$ and $C$ are the number of connected components of $\mathcal{G}_r$ and $\mathcal{G}_c$ respectively.

\begin{proposition}
\label{prop:dof_monotonicity}
If \As{missingness} holds, then the degrees of freedom possesses the following properties: (i) $\tr(\M{S}\Inv(\V{\gamma})) \geq RC$ for all $\V{\gamma}$ positive, (ii) $\tr(\M{S}\Inv(\V{\gamma})) \geq \tr(\M{S}\Inv(\V{\gamma}'))$ whenever $\gamma'_r \geq \gamma_r$ and $\gamma'_c \geq \gamma_c$, and (iii) $\tr(\M{S}\Inv(\Vn{\gamma}{k})) \rightarrow RC$ for any sequence $\Vn{\gamma}{k}$ such that $\underline{\gamma}^{(k)} \rightarrow \infty$.
\end{proposition}

The BIC for BMC is given by the following expression.
\begin{eqnarray*}
\bic(\gamma_r,\gamma_c) & = & \lvert \Omega \rvert\log \left (\Vert\mathcal{P}_{\Omega}(\M{x}) - \mathcal{P}_{\Omega}(\M{z}) \Vert_{\text{F}}^2 \right ) + \log(\lvert \Omega \rvert)\text{df}. 
\end{eqnarray*} 
We can re-express the BIC in terms of $\M{S}$ to make the dependence of the BIC on $\gamma_r$ and $\gamma_c$ more explicit.
\begin{eqnarray*}
\bic(\gamma_r,\gamma_c) & = & \lvert \Omega \rvert\log \left (\Vert\M{P}_{\Omega}\V{x} - \M{P}_{\Omega}\M{S}\Inv\V{x}) \Vert_{\text{F}}^2 \right ) + \log(\lvert \Omega \rvert)\tr(\M{S}\Inv).
\end{eqnarray*} 

A naive approach fits models over a grid of $(\gamma_r, \gamma_c)$ values and chooses the pair of regularization parameters that minimize the BIC. But this approach does not leverage the differentiability of the BIC. Since the BIC is differentiable, we compute the gradient with respect to $\V{\gamma}$ and employ gradient descent. Let $\V{z}$ denote the solution of the system $ \M{S}(\gamma_r, \gamma_c)\V{z} = \M{P}_\Omega \V{x}$ and define the residual $\V{r} \equiv \M{P}_\Omega(\V{z}-\V{x})$. Then the gradient is given by the following equation.
\begin{eqnarray}
\label{eq:gradient}
\nabla \bic(\gamma_r, \gamma_c) & = &
\begin{pmatrix}
-\frac{2\lvert \Omega \rvert}{\lVert \M{P}_\Omega \V{r} \rVert^2}\V{x}\Tra\M{P}_\Omega \M{S}_r\M{P}_\Omega \V{r} - \log(\lvert \Omega \rvert) \tr( \M{S}_r)\\
-\frac{2\lvert \Omega \rvert}{\lVert \M{P}_\Omega \V{r} \rVert^2}\V{x}\Tra\M{P}_\Omega \M{S}_c\M{P}_\Omega \V{r} - \log(\lvert \Omega \rvert) \tr(\M{S}_c)
\end{pmatrix},
\end{eqnarray}
where $\M{S}_r \equiv \M{S}\Inv  (\M{I}\Kron \M{L}_r) \M{S}\Inv$ and $\M{S}_c \equiv \M{S}\Inv  (\M{L}_c\Kron \M{I}) \M{S}\Inv$. We provide a derivation in the Supplementary Materials. 
With the gradient in hand, for little additional cost we may employ accelerated first order methods, for example SparSa \citep{Wright2008} or FISTA \citep{Beck2009}, or a Quasi-Newton method \citep[Ch 6]{Nocedal2006}. In this article, we apply the Quasi-Newton method. Note that with a trivial modification we can extend the IMS strategy to iteratively minimizing the AIC.

\paragraph{Computational Complexity} The key cost in evaluating the BIC is computing the trace of $\M{S}^{-1}$ and for \Eqn{gradient} the key cost is computing the trace of $\M{S}_c$ and $\M{S}_r$. Here, one can no longer take advantage of the sparsity of the system, since the inverse of $\M{S}$ is dense and therefore the computational complexity of computing the inverse is now $\mathcal{O}((np)^3)$. The same is also true for computing $\M{S}_c^{-1}$ and $\M{S}_r^{-1}$. Moreover, storing a dense matrix is infeasible when $n$ and $p$ are large.
 
Grid based model selection using a BIC or AIC will require $\mathcal{O}(N_rN_c(np)^{3})$ work, as each grid point requires trace of the inverse of an $np$-by-$np$ matrix which can be prohibitively expensive for large-scale problems of interest. We tackle this computational challenge with a two-pronged strategy: (i) we approximate the trace computation using a Monte Carlo method, to reduce the cost of each objective function evaluation, and (ii) we use an optimization approach for minimizing the BIC. 



\section{Scaling up Iterative Model Selection}
\label{sec:monte_carlo}

In this section, we briefly review Monte Carlo estimators for computing the trace of a matrix and discuss the computational costs associated with it. Next, we express the BIC minimization problem as a stochastic programming problem; we then approximate the objective function using a Monte Carlo estimator. This is called the Sample Average Approximation (SAA) method. 

\subsection{Monte Carlo Trace Estimator}

In the application at hand, $\M{S}$ is a large {\em sparse} matrix; however, $\M{S}\Inv$ is a large {\em dense} matrix. Consequently, forming $\M{S}\Inv$ explicitly is not advisable. We turn to a matrix-free approach for estimating the trace. \citet{hutchinson1989stochastic} introduced an unbiased estimator for the trace of a positive semidefinite matrix $\M{M} \in \mathbb{R}^{m\times m}$,
\begin{eqnarray}
\label{eq:Hutchinson}
\tr(\M{M}) &  = & \mathbb{E}_{\V{w}} [\V{w}\Tra\M{M}\V{w}] \amp \approx \amp \frac{1}{N}\sum_{k=1}^{N} \V{w}_k\Tra\M{M}\V{w}_k \,.
\end{eqnarray}
Here $\V{w}_k \in \Real^m$ are i.i.d.\ samples from a Rademacher distribution, i.e.\@, $\V{w}_k$ has entries $\pm 1$ with equal probability.  Other choices for distributions have been proposed \citep{avron2011randomized}; in general, $\V{w}_k$ must have zero mean and identity covariance. Examples of alternative distributions include Gaussian random vectors and random samples from columns of unitary matrices \citep{avron2011randomized,roosta2015improved}.
 
The quality of the Monte Carlo estimator given in \Eqn{Hutchinson} can be judged by two different metrics -- variance and the number of samples for an $(\epsilon, \delta)$-estimator.  \cite{hutchinson1989stochastic} showed that the variance of the estimator given in \Eqn{Hutchinson} is $2\left(\lVert \M{M} \rVert_{\text{F}}^2 - \sum_{j=1}^{n}\ME{M}{jj}^2\right)$. Therefore, if the off-diagonal entries of the matrix are large compared to its diagonal entries, the variance of the Hutchinson estimator can be quite large. An estimator is called an $(\epsilon,\delta)$-estimator (for $\epsilon > 0$ and $\delta < 1$) if the relative error of the trace is less than $\epsilon$ with probability at least $1 - \delta$.~\cite{roosta2015improved} showed that the minimum number of samples from a Rademacher distribution for an $(\epsilon,\delta)$-estimator is $6\epsilon^{-2}\log(2/\delta)$. This result implies that for an accurate estimator (i.e.\@, small $\epsilon)$, many samples are required ($\propto \epsilon^{-2}$). However, as we will demonstrate in the numerical experiments, a modest number of samples $\lesssim 10$ are sufficient for our purpose.   

Employing the Hutchinson estimator requires repeated evaluation of quadratic forms $\V{w}_k\Tra\M{S}\Inv\V{w}_k$, in order to estimate the $\tr(\M{S}\Inv)$. The inverse $\M{S}\Inv$ need not be computed explicitly; instead, we first solve $\M{S}\V{z}_k = \V{w}_k$, and then compute the inner product $\V{w}_k\Tra\V{z}_k$. Methods to solve linear system are discussed in Section~\ref{sec:estimation}. 

\subsection{Sample Average Approximation (SAA) method}
\label{eq:SAA}
Following~\cite{Anitescu2012}, we substitute the exact trace operation in the objective function with its Hutchinson estimate, to obtain a stochastic programming problem. Define
 \begin{eqnarray*}
f(\V{\gamma}, \V{w}) & \equiv & \lvert \Omega \rvert \log \left (\Vert\M{P}_{\Omega}(\M{I} - \M{S}\Inv)\M{P}_{\Omega}\V{x} \Vert_2^2 \right) + \log(\lvert \Omega \rvert)\ \V{w}\Tra\M{S}\Inv\V{w}.
\end{eqnarray*}
Therefore, the minimizer of the BIC is equivalent to the solution of the optimization problem  
$\min_{\V{\gamma}}\> \mathbb{E}_{\V{w}} [f(\V{\gamma}, \V{w})]$.

The SAA to the objective function gives an unbiased estimator of the $\bic$ function 
\begin{equation}
\label{eq:bic_saa}
\begin{split}
\widehat{\bic}(\V{\gamma},\V{w}_1, \ldots, \V{w}_N) & = \frac{1}{N} \sum_{i=1}^N f(\V{\gamma}, \V{w}_i) \\
& = \lvert \Omega \rvert \log \left (\Vert\M{P}_{\Omega}(\M{I} - \M{S}\Inv)\M{P}_{\Omega}\V{x} \Vert_2^2 \right) + \log(\lvert \Omega \rvert)\frac{1}{N}\sum_{i=1}^N \V{w}\Tra_i\M{S}\Inv \V{w}_i,
\end{split}
\end{equation} 
where $\V{w}_i$ are i.i.d.\ random vectors in $\Real^{np}$ with mean zero and identity covariance. Following the previous subsection, we choose $\V{w}_i$ as Rademacher random matrix.  
As before, we compute the gradient of \Eqn{bic_saa}. 

A natural question, is what is the relationship between the regularization parameters $\V{\gamma}_N$ obtained by minimizing the approximate $\widehat{\bic}$ to regularization parameters $\V{\gamma}^\star$ obtained by minimizing the true \bic. Let $\V{\gamma}^\star$ denote a stationary point of $\bic(\V{\gamma})$, namely $\nabla \bic(\V{\gamma}^\star) = \V{0}$ and suppose further that the Hessian of the BIC at $\V{\gamma}^\star$ is nonsingular, namely $\nabla^2 \bic(\V{\gamma}^\star)$ is nonsingular. Let $\V{\gamma}_N$ denote a stationary point to the unbiased estimate of the \bic, namely
\begin{eqnarray*}
\nabla_\V{\gamma} \widehat{\bic}(\V{\gamma}_N(\omega),\V{w}_1(\omega), \ldots, \V{w}_N(\omega)) & = & \V{0},
\end{eqnarray*}
where we have made explicit the dependency of $\V{\gamma}_N$ on the outcome $\omega$ to emphasize that $\V{\gamma}_N$ is a random variable. Then with the assumptions of~\citet[Theorem 5.14]{shapiro2014lectures}  with probability 1, the stationary point of $\widehat{\bic}$ is locally unique in a neighborhood of $\V{\gamma}^\star$ for $N$ sufficiently large. 

Furthermore, define
\begin{eqnarray*}
\M{J}_N & = & \frac{1}{N}\sum_{i=1}^N \nabla^2_{\V{\gamma}} f(\V{\gamma}_k, \V{w}_i), \\
\M{\Sigma}_N & = & \frac{1}{N}\sum_{i=1}^N \nabla_{\V{\gamma}} f(\V{\gamma}_N, \V{w}_i) \nabla_{\V{\gamma}} f(\V{\gamma}_N, \V{w}_i) \Tra,
\end{eqnarray*}
and the empirical covariance matrix $\M{V}_N  =  \M{J}_N\Inv \M{\Sigma}_N \M{J}_N\Inv$.
From~\citep[Theorem 5.14]{shapiro2014lectures}, it can be shown that ${N}^\frac{1}{2}\M{V}^{-\frac{1}{2}}_N(\V{\gamma}_N - \V{\gamma}^\star) \overset{\mathcal{D}}{\rightarrow} \mathcal{N}(\V{0}, \M{I}).$ 

\subsection{Summary of Computational Costs for Model Selection}

We now compare and contrast the computational costs of the various strategies for performing model selection in the BMC problem.
Direct grid search to obtain model selection parameters using $K$-fold cross-validation is expensive both in terms of the per function evaluation $\mathcal{O}(K(np)^{3/2})$ and number of function evaluations $N_rN_c$ which amounts to a total computational cost of $\mathcal{O}(KN_rN_c(np)^{3/2})$. The cost for direct grid search using the BIC is even worse; evaluating the BIC function over $N_r N_c$ grid points costs $\mathcal{O}(N_rN_c(np)^3)$. Using the Hutchinson approximation, however, substantially lowers the cost of the BIC function evaluation, because now we need to solve a sparse system rather than explicitly invert the matrix. The number of function evaluations, however, remains the same and therefore, the computational cost is $\mathcal{O}(NN_rN_c(np)^{3/2})$, where $N$ is the number of vectors used in the Monte Carlo approximation. In the IMS, a Quasi-Newton approach is used to optimize the BIC to obtain the model parameters, where now, the objective function and the gradient have been approximated using the Hutchinson trace approach. The objective function costs $\mathcal{O}(N(np)^{3/2})$, and the gradient evaluation requires an additional cost of $\mathcal{O}(N(np)^{3/2})$. We note that the gradient computation only requires solving only one additional linear system involving $\M{S}$, since intermediate computations involving the objective function can be reused for evaluating the gradient, details are provided in the Supplementary Materials. It is clear that the IMS is computationally cheaper than CV if the number of IMS iterations $N_\text{IMS}$ satisfies $N_\text{IMS}N \ll N_rN_cK$. This is indeed what we will see in our numerical experiments in \Sec{Examples}. A summary of the computational costs is provided in Table~\ref{t_costs}. In the above analysis, we have assumed a direct solver has been used to solve systems involving $\M{S}$; an iterative solver may be computationally beneficial for large-scale systems and the cost is similar. 
\begin{table}[h]\centering
\begin{tabular}{c|cccc}
 & CV  & BIC &  BIC + Hutchison & IMS \\ \hline
Obj. func. & $\mathcal{O}(KN_rN_c(np)^{3/2})$ &  $\mathcal{O}(N_rN_c(np)^{3})$&  $\mathcal{O}(N_rN_cN(np)^{3/2})$ & $\mathcal{O}(NN_\text{IMS}(np)^{3/2})$ \\
Gradient & - & - & - & $\mathcal{O}(N_\text{IMS}(np)^{3/2})$ \\
\end{tabular}
\caption{Comparison of computational cost of different approaches. `CV' refers to $K$-fold cross-validation, `BIC' refers to BIC grid search, `BIC + Hutchinson' refers to BIC grid search with BIC approximated using the Hutchinson trace approximation. `IMS' refers to Iterative Model Selection using a combination of Quasi-Newton method and Hutchinson approximation to the objective function and gradient.} 
\label{t_costs}
\end{table}

\section{Numerical Experiments}
\label{sec:Examples}

We now discuss numerical experiments to evaluate the exact and approximate IMS methods on simulated and real data. We also compare the IMS methods to standard grid-search strategies. All experiments were conducted in Matlab. To compare timing results, we record the time in seconds required to complete model selection. We perform computations in serial on a multi-core computer with twenty four 3.3 GHz Intel Xeon processors and 189 GB of RAM.  

\subsection{Simulated Data}
\label{sec:simulated_data}

We first compare IMS via Quasi-Newton optimization, BIC grid-search, and cross-validated grid-search on simulated data. We consider two versions of the Quasi-Newton optimization: 
(i) exact computation and (ii) Hutchinson estimation.
Identical experiments with the AIC used in place of the BIC lead to similar results and are summarized in the Supplementary Materials.

In all simulated data experiments, the matrix $\M{M}$ that we seek to recover consists of four biclusters.
\begin{eqnarray*}
\M{M} & = & 
\begin{pmatrix}
		10\M{J}_{25} & -25\M{J}_{25}\\
		25\M{J}_{25} & -10\M{J}_{25} \\
		\end{pmatrix},
\end{eqnarray*}
where $\M{J}_q$ is the $q\times q$ matrix of all ones. 
We observe the noisy matrix $\M{X} = \M{M} + \M{\mathcal{E}}$, where $\varepsilon_{ij}$ are i.i.d.\ $\mathcal{N}(0,1)$. We use the following row and column weights
\begin{eqnarray*}
\ME{W}{ij} \amp = \amp \tilde{w}_{ij} & = &
	\begin{cases}
		0 & \text{if} \ i=j\\
		1 & \text{if} \ i\neq j \ \text{and} \ i, j \in C_1 \ \text{or}  \ i, j \in C_2\\
		0.001 & \text{if} \ i\neq j \ \text{and} \ i \in C_1 , j \in C_2 \ \text{or} \ i \in C_2 , j \in C_1
	\end{cases},
\end{eqnarray*}
where $C_1 = \{1, \dots, 25\}$ and $C_2 = \{26, \dots, 50\}$. Thus, the weights introduce some erroneous smoothing across distinct biclusters. While the noise level and weights choices are admittedly not particularly challenging from an inferential perspective, our main objective in these studies is to understand the computational impact of the choices we make in deciding upon a model selection procedure.

\begin{figure}[t]
    \centering
	\includegraphics[scale=0.49]{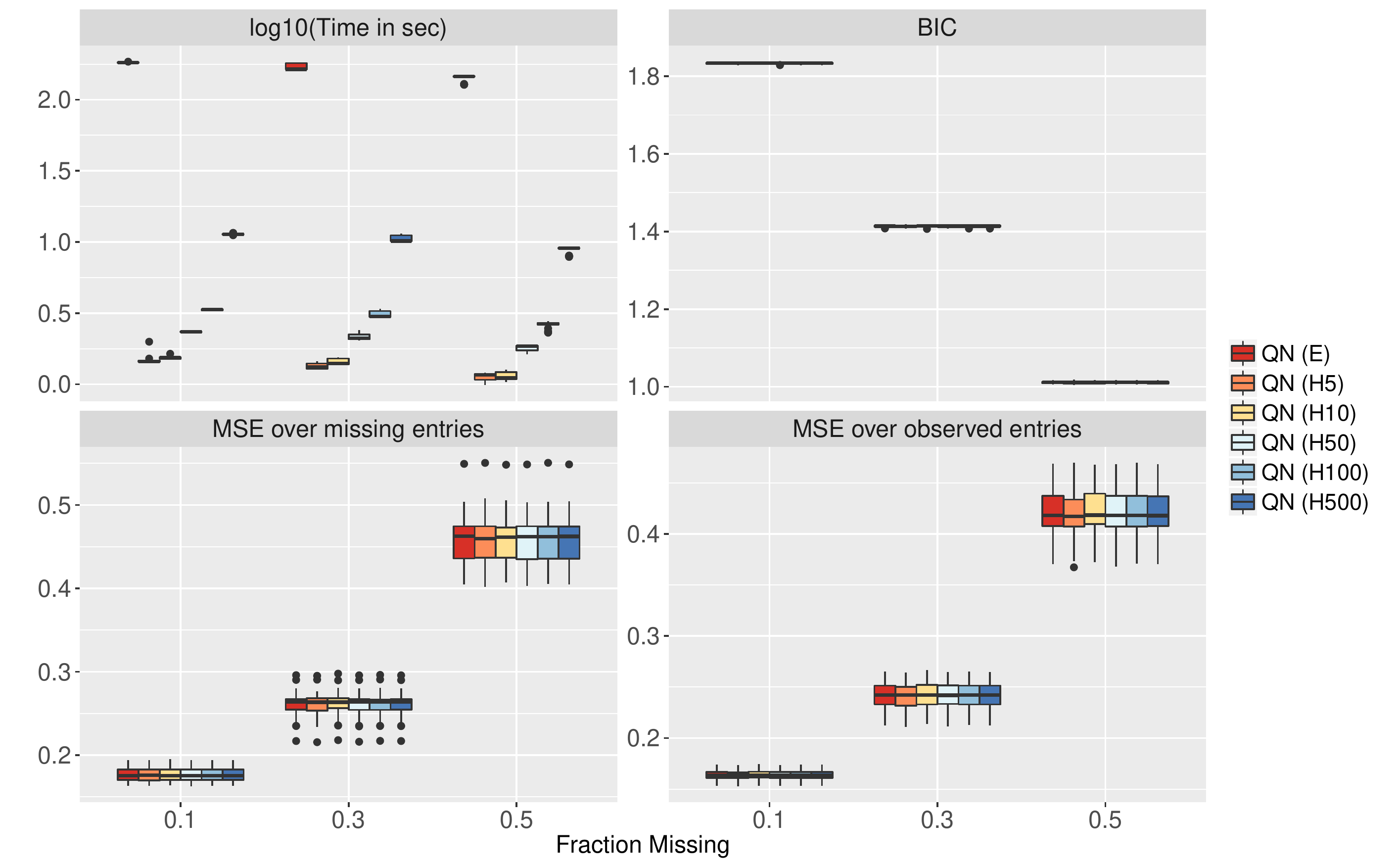}
	\caption{Comparison between IMS via Quasi-Newton with exact computation (E) and IMS via Quasi-Newton with Hutchinson estimation (HN indicates $N$ samples), under different missing fractions. \label{fig:numerical_ex4}}
\end{figure}

We perform three different experiments. The first experiment evaluates the run-time versus accuracy tradeoff of IMS via the Quasi-Newton method when using the exact computation versus the Hutchinson estimator. The point of this experiment is to assess how gracefully the quality of the stochastic approximation of the BIC degrades as a function of the number of samples used to compute the approximation. The second and third experiments compare the run-time versus accuracy tradeoff of the Quasi-Newton method against two standard grid-search model selection methods: cross-validation grid-search and grid-search over the BIC surface. For all experiments, we use three different missingness fractions (0.1, 0.3 and 0.5); metrics are averaged over 30 replications.

\begin{figure}[t]
    \centering
	\includegraphics[scale=0.49]{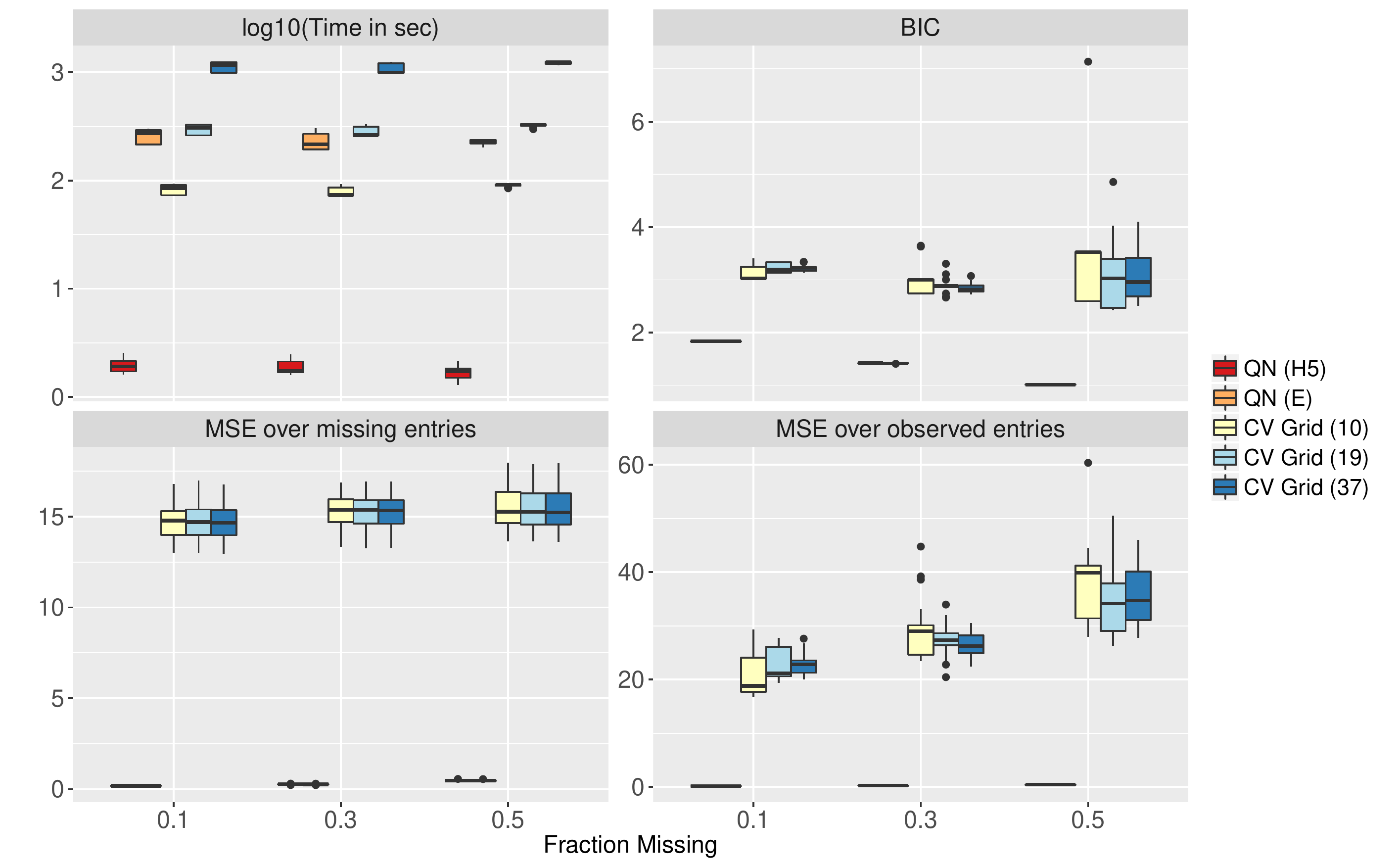}
	\caption{Comparison of (i) IMS via Quasi-Newton with exact computation (E), (ii) IMS via Quasi-Newton with Hutchinson estimation (HN indicates $N$ samples), and (iii) cross-validation grid-search, under different missing fractions.\label{fig:numerical_ex5}}
\end{figure}

\Fig{numerical_ex4} compares exact computation and Hutchinson estimation in terms of runtime, BIC, and mean squared error (MSE) over missing entries and observed entries. Unexpectedly, regardless of the method, the prediction accuracy increases as the fraction of missing entries decreases. Remarkably, however, the Quasi-Newton method with Hutchinson estimation can recover the matrix as well as the Quasi-Newton method with exact computation even when the sample size is 5. This is a nontrivial windfall, as using just 5 samples takes significantly less time than the exact computation. In light of this result, we use a sample size of 5 whenever we employ the Hutchinson estimator in subsequent experiments with simulated data.

\begin{figure}[t]
    \centering
	\includegraphics[scale=0.49]{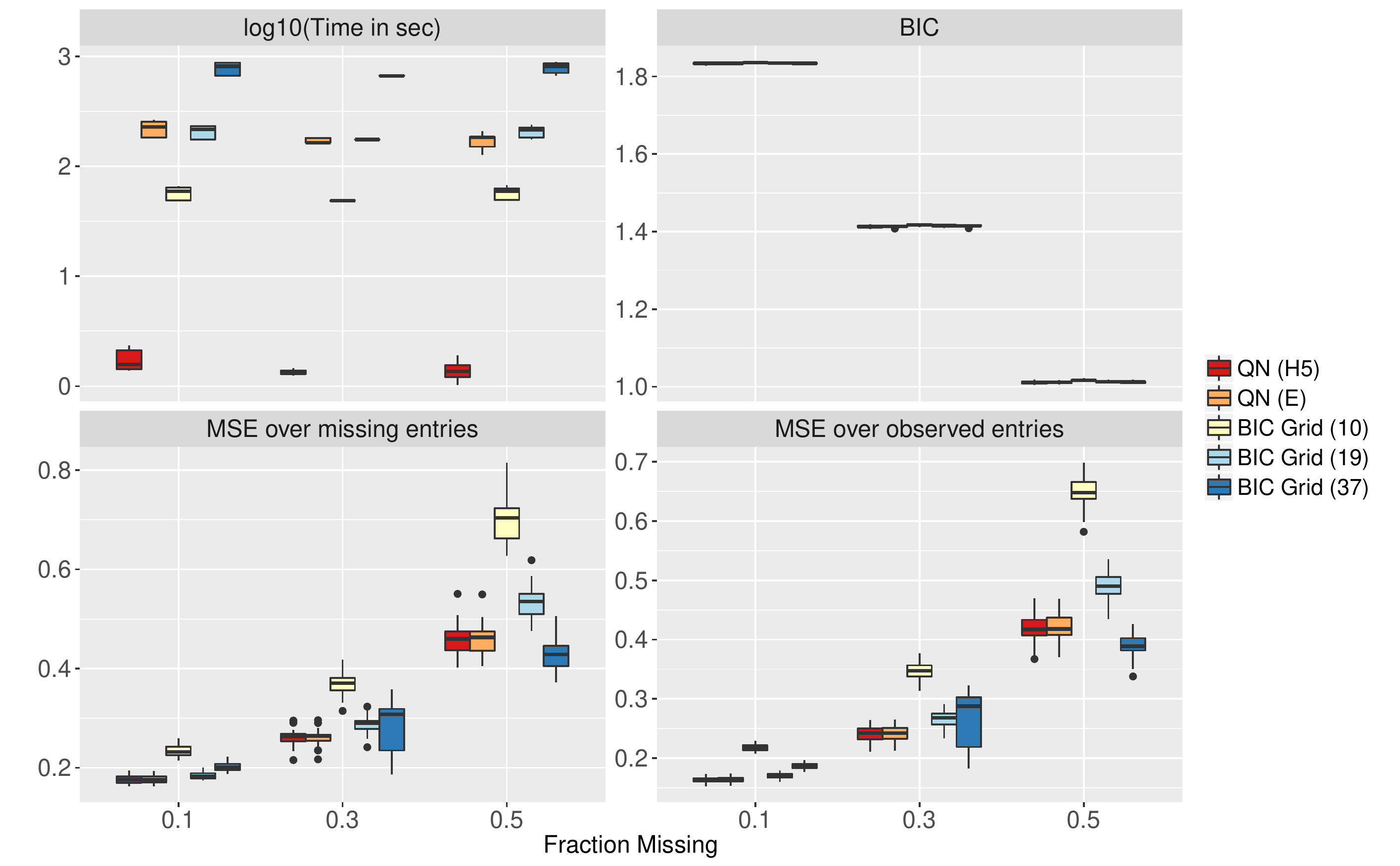}
	\caption{Comparison of (i) IMS via Quasi-Newton with exact computation (E), (ii) IMS via Quasi-Newton with Hutchinson estimation (HN indicates $N$ samples), and (iii) BIC grid-search, under different missing fractions.\label{fig:numerical_ex6}}
\end{figure}

\Fig{numerical_ex5} compares the Quasi-Newton method (exact computation and Hutchinson estimation) against cross-validation grid. We use 5-fold cross-validation on MSE over missing entries, and we test on three different levels of grid coarseness and three missingness fractions. All grid-searches occur over the range $(\gamma_r, \gamma_c) \in [e^{-9},e^1] \times [e^{-9},e^1]$. We denote by Grid ($N$) the set of $N$ evenly spaced points on the interval $[-9,1]$ that are then exponentiated. Thus, larger $N$ corresponds to a finer grid. The upper left and right panels in \Fig{numerical_ex5} show that Quasi-Newton with Hutchinson estimation takes the least time and has the best performance in objective value BIC. The lower left and right panels in \Fig{numerical_ex5}, show that the parameters chosen using BIC criteria lead to models with better performance in MSE than those chosen using cross-validation on MSE. The ability to go off the parameter grid is evident. Even the finest grid-search, Grid (37), in our study cannot reach the optimal MSE achieved obtained via the IMS via Quasi-Newton methods.

\Fig{numerical_ex6} compares the Quasi-Newton method (exact computation and Hutchinson estimation) against BIC grid-search with different levels of coarseness.  The upper left panel in \Fig{numerical_ex6} again shows that the Quasi-Newton direct optimization takes less time than searching over a finer grid, and the Quasi-Newton with Hutchinson estimation takes even less time. The lower left and right panels in \Fig{numerical_ex6} show that the Quasi-Newton with Hutchinson method achieves lower prediction error than grid-search when 10\% of the entries missing. Even though the finest grid-search, Grid (37), has better performance on average, when greater fractions of data are missing, the Quasi-Newton methods are not far behind. Employing the Quasi-Newton method with Hutchinson estimation is clearly attractive given its accuracy and superior run time.

\subsection{Real Data Example from Radiogenomics}
\label{sec:radiogenomics}

The goal in radiogenomics is to create a rational set of rules for recommending a patient for genomic testing based on a collection of radiographic findings \citep{Rutman2009, Colen2014}. The key task is to identify associations between radiographic features and genomic pathway activity. In the case of glioma radiology, the Visually Accessible Rembrandt Images (VASARI) is a standard way of reporting MRI finding for gliomas \citep{Gutman2013}. In addition, computational approaches to identifying image features based on tumor volume and variation in voxel intensities are also used \citep{Yang2015}.

To set some notation, suppose on $n$ patients we obtain two sets of measurements: a matrix $\M{U} \in \Real^{n\times p}$ where the $i$th row $\V{u}_i$ is the vector of radiographic features for the $i$th patient and a matrix $\M{V} \in \Real^{n\times q}$ of gene expressions where the $i$th row $\V{v}_i$ is the vector of pathway activities. To relate these computationally-derived image features with gene expression data, we consider the cross-covariance matrix $\M{X} = \M{U}\Tra\M{V}$. The $ij$th entry $\VE{x}{ij}$ of $\M{X}$ is the quantifies the association between the $i$th imaging feature and the $j$th pathway. The objective is to identify correlated features such as tumor size with gene mutations and ultimately derive more principled rules for ordering genetic testing.

There is a missing data challenge, however, as patients may be missing annotation on some radiographic features and gene pathway expression levels. Nonetheless, similarity weights can be inferred for the radiographic features as well as for the gene pathways using measurements from different modalities. The availability of similarity structure suggests that the problem of identifying missing associations in a radiogenomics study may be accomplished using BMC. We now consider how BMC performs on a radiogenomics example involving a subset of patients from The Cancer Genome Atlas (TCGA).

\begin{figure}[t]
	\centering
	\includegraphics[scale=0.35]{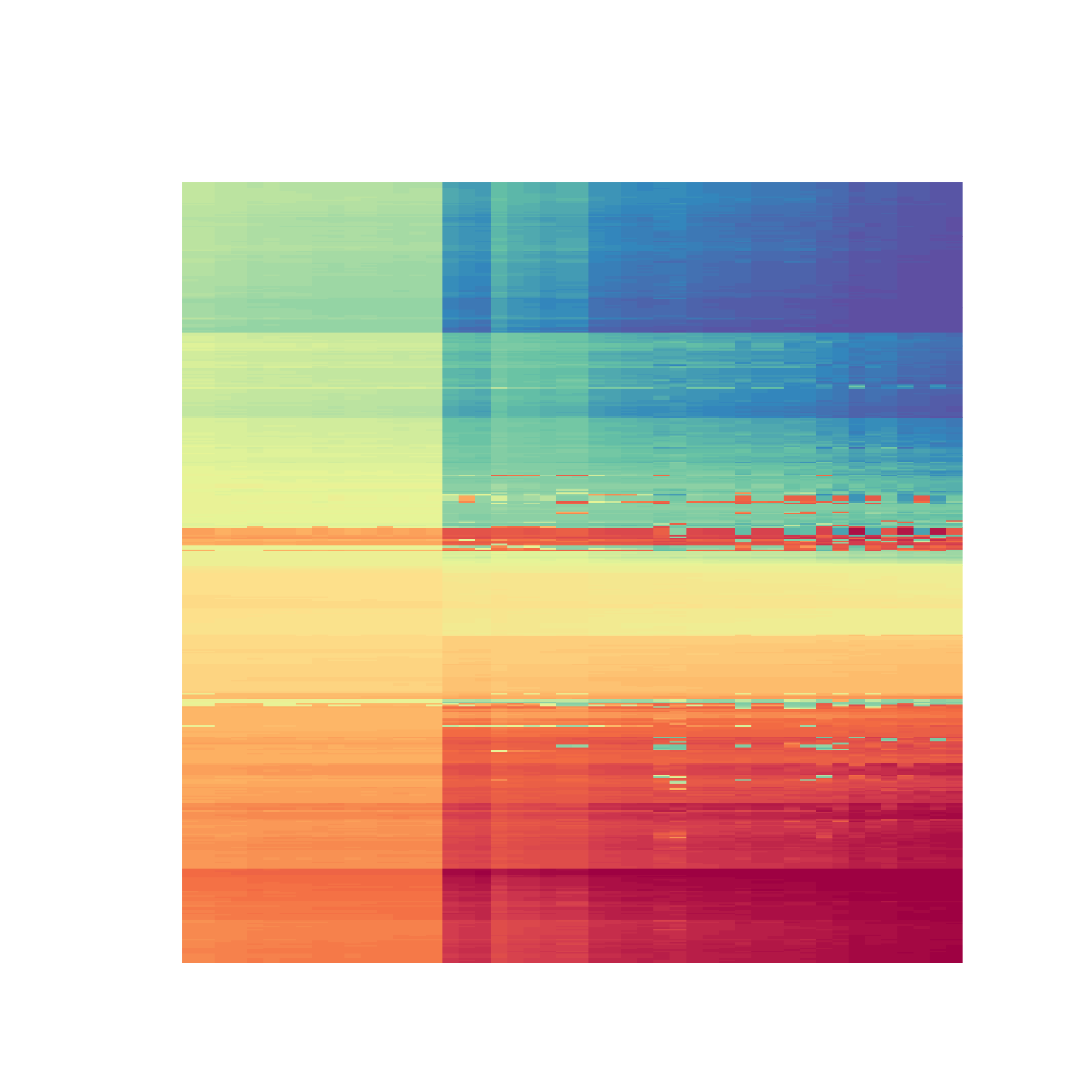}
	\caption{The complete cross covariance matrix between 48 SFTA texture features extracted from T1-post contrast MRI scans and 533 pathways.}
	\label{fig:radgeno2_crosscov}
\end{figure}

\Fig{radgeno2_crosscov} shows the cross covariance matrix between 48 segmentation-based fractal texture (SFTA) features extracted from T1-post contrast MRI scans and 533 genomic pathways\footnote{Pathway data is available at \href{https://gdac.broadinstitute.org/}{https://gdac.broadinstitute.org/}.}. The 48 SFTA features were obtained by using the method of \citet{Costa2012}. 
Both imaging features and pathways were recorded on the same set of 77 TCGA patients. Rows and columns have been reordered using single-linkage hierarchical clustering on the rows and columns independently. Reordering the rows and columns reveals that the data has a checkerboard pattern.

We construct weights in two stages analogous to the construction of $k$-nearest-neighbor graphs in spectral clustering \citep{Luxburg2007}. We describe how row weights are constructed; columns weights are construction similarly. Initial row weights $w_{ij}$ consist of the exponentiated Pearson correlation between the row. Thus, only positive correlations led to strong smoothing shrinkage. Although a more sophisticated weight choice may take advantage of negative correlations, these simple weights are effective for our purposes, as our focus is on evaluating the computational performance of IMS. We then make the row weights sparse, or mostly zero, using the following rule. Fixing $i$, we find the 5 largest values of $w_{ij}$. If $j$ is not among these top values, we set $w_{ij} = 0$. We then repeat this step with $j$ fixed. We do this for all $i$ and $j$. Approximations to this procedure should be employed and can be accomplished with approximate $k$-nearest-neighbors algorithms \citep{SlaCas2008}, as searching over all pairs requires computation that grows quadratically in the number rows. 

\begin{figure}[t]
    \centering
	\includegraphics[scale=0.49]{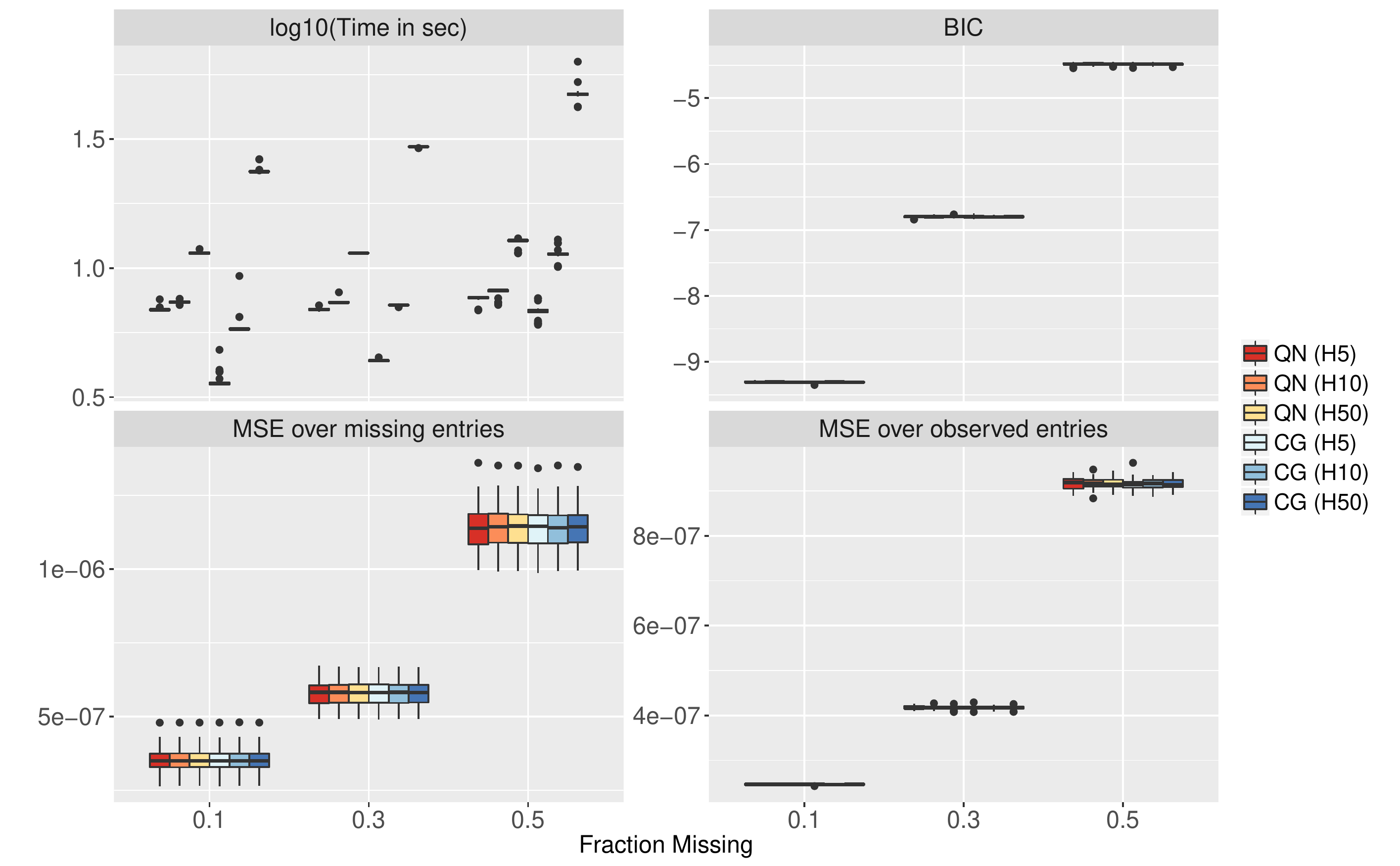}
	\caption{Radiogenomics Data: Comparison between Quasi-Newton with Hutchinson estimation (QN, size=N) and conjugate gradient with Hutchinson estimation (CG, size=N), under different sample size ($N=5,10,50$) and different missing fractions. \label{fig:radgeno}}
\end{figure}

Since the cross-covariance matrix is larger than the simulated matrices considered earlier, we performed the following illustrative experiment to compare the performance using Quasi-Newton with Hutchinson estimation and conjugate gradient (CG) with Hutchinson estimation. It is important to note here that the CG method is used to solve a linear system in order to compute the gradient of the BIC and not used as a method to minimize the BIC. Results are shown in \Fig{radgeno}. The experiment were repeated for 30 times under 3 different missing fractions: 0.1, 0.3 and 0.5. All methods produce equally good predictions at all levels of missingness. We see that for the larger matrix, however, that the CG method can provide some additional speed ups for larger matrices such as the radiogenomics cross-covariance considered here.

\section{Discussion}
\label{sec:conclusion}

In this paper, we revisited the matrix completion problem when we have additional information on row and column similarities. We posed this problem as a penalized convex least squares problem and established several properties of this problem formulation, in particular when this problem admits a unique solution. We also introduced an efficient algorithm called IMS for tuning the amount of regularization in practice. We showed that when rows and columns exhibit a strong clustering pattern, a pair of differentiable Laplacian penalties can recover a low-rank structure that is consistent with the row and column similarities. This motivates solving a differentiable convex optimization problem, that has been previously proposed in the literature, with two penalty terms instead of a nondifferentiable convex optimization problem with three penalty terms one of which is a nuclear norm penalty. Dropping the nuclear norm penalty has three advantages: (i) an expensive SVD calculation is avoided, (ii) model selection is reduced to searching for two tuning parameters, instead of three, and (iii) model selection via the BIC can be achieved by minimizing a differentiable function of two variables. We emphasize that what makes advantage (iii) possible is that the degrees of freedom has a differentiable closed form expression. If we included the nuclear norm penalty, we could derive an unbiased estimate of the degrees of freedom following \citep{Candes2013}, but the resulting estimate would not be differentiable and consequently the BIC could not be minimized via the Quasi-Newton method.

Exhaustively searching for a minimizer of a surrogate measure of predictive performance over a regular two-way grid of tuning parameters is typically inefficient. Ideally one would place grid points more densely in the model space where the optimal model resides. Unfortunately, this information is not known in advance. 
With IMS, we do not need to construct a grid of candidate regularization parameters and can even identify better predictive models by going off the parameter grid and not limiting the search to models defined by the parameters on the grid. Since the work required to fit a model at a grid point is essentially the same as the work required for a single IMS step, searching the space with IMS can lead to drastically fewer linear system solves in practice.

We can further expedite model selection by reducing the amount of work expended at each gradient calculation. In this paper, we proposed using the Hutchinson estimator to approximate the trace term in the BIC. This stochastic approximation to the BIC is then minimized using the Quasi-Newton method. Remarkably, our numerical experiments demonstrated that even coarse approximations regularly lead to the selection of models with prediction errors that rivaled those obtained by minimizing the exact BIC. This is significant because computations with the approximations take an order of magnitude less time than their exact counterparts.

As mentioned earlier, alternatives to the Quasi-Newton method could be employed in IMS. We leave it as future work to investigate how alternative modified first order methods might fare against the Quasi-Newton approach proposed here. We also leave it as future work to investigate how other stochastic approaches, such as stochastic Quasi-Newton method \citep{Byrd2016}, might fare against the SAA method proposed here.

Exploiting the special structure in this version of the matrix completion problem can lead to surprisingly effective computational gains in model selection. We close by noting that this simple but effective strategy applies in contexts outside of matrix completion and should be considered as an alternative to grid-search methods for automating and streamlining model selection when two or more tuning parameters need to be identified. 


\newpage

\appendix
\section{Verifying \As{missingness} in work linear in the data}
\label{sec:verify}

Recall that the connected components of a graph can be determined in work that is linear in the number of vertices via a breadth-first-search. The brief description of the algorithm, along with the computational costs, is provided in Algorithm~\ref{alg:verify}

\begin{algorithm}[!ht]
Initialize $\M{M} \in \Real^{n \times p}$: Set $\ME{M}{ij} = 1$ if $(i,j) \in \Omega$ and 0 otherwise.
\begin{algorithmic}[0]
  \caption{Check \As{missingness}}
  \label{alg:verify}
\State $(A_1, \ldots, A_R) \gets$ Find-Connected-Components($\mathcal{G}_r$) 
\Comment $\mathcal{O}(n)$ work
\State $(B_1, \ldots, B_C) \gets$ Find-Connected-Components($\mathcal{G}_c$) 
\Comment $\mathcal{O}(p)$ work
\For{$r = 1, \ldots, R$ and $c = 1, \ldots, C$}
\Comment $\mathcal{O}(np)$ work
\If{ $\V{\chi}_{A_r}\Tra \M{M} \V{\chi}_{B_c} = 0$}
\State \Return False
\EndIf
\EndFor
\State \Return True
\end{algorithmic}
\end{algorithm}

\section{Proofs of Results in \Sec{bmc_solution} and \Sec{model_selection}}
\label{sec:proof}

In this section we give proofs of the propositions and theorems within the paper.

\subsection{\Prop{shrinkage}}

We first recall a key fact about the number of connected components of a graph and the spectrum of its graph Laplacian matrix.

\begin{proposition}[Proposition 2 in \cite{Luxburg2007}]
\label{prop:spectral}
Let $\mathcal{G}$ be an undirected graph with non-negative weights. Then the multiplicity $k$ of the eigenvalue 0 of $\M{L}$, the graph Laplacian of $\mathcal{G}$, equals the number of connected components $A_1, \ldots, A_k$ in the graph $\mathcal{G}$. The eigenspace of eigenvalue 0 is spanned by the indicator vectors $\V{\chi}_{A_1}, \ldots, \V{\chi}_{A_k}$ of those components.
\end{proposition}

We are now ready to prove \Prop{shrinkage}.
\begin{proof}
First assume that $\M{Z} = \sum_{r=1}^R\sum_{c=1}^C \mu_{rc}\V{\chi}_{A_r}\V{\chi}_{B_c}\Tra$ for some $\mu_{rc}$ for $r =1, \ldots, R$ and $c = 1, \ldots, C$.
\Prop{spectral} implies that $\M{L}_r\V{\chi}_{A_r} = \V{0}$ for all $r = 1, \ldots, R$. Now invoke the linearity and cyclic permutation properties of the trace function to simplify the expression $\tr(\M{Z}\Tra\M{L}_r\M{Z})$.
\begin{eqnarray*}
\tr(\M{Z}\Tra\M{L}_r\M{Z}) \amp = \amp \sum_{r=1}^R\sum_{c=1}^C \mu_{rc} \tr(\M{Z}\Tra\M{L}_r\V{\chi}_{A_r}\V{\chi}_{B_c}\Tra) \amp = \amp 
\sum_{r=1}^R\sum_{c=1}^C \mu_{rc} \langle \M{Z}\V{\chi}_{B_c}, \V{0} \rangle 
 \amp = \amp 0.
\end{eqnarray*}
Analogously, $\tr(\M{Z}\M{L}_c\M{Z}\Tra) = 0$, and consequently $J(\M{Z}) = 0$.

Now assume that $J(\M{Z}) = 0$. 
Let $\mathcal{F}$ denote the set of rank-1 matrices defined by the connected components of $\mathcal{G}_r$ and $\mathcal{G}_c$, namely $\mathcal{F} = \{ \V{\chi}_{A_r}\V{\chi}_{B_c}\Tra : r = 1, \ldots, R, c = 1, \ldots, C\}$. The set $\mathcal{F}$ is a basis for the $RC$ dimensional subspace of matrices span$(\mathcal{F})$. Let $\mathcal{F}^\perp$ denote the orthogonal complement of $\mathcal{F}$ and suppose that $\M{Z} = \M{M} + \M{N}$ where $\M{M} \in \mathcal{F}$ and $\M{N} \in \mathcal{F}^\perp$. Consequently, $\tr(\M{Z}\Tra\M{L}_r\M{Z}) = \tr(\M{N}\Tra\M{L}_r\M{N}) = 0$ and $\tr(\M{Z}\M{L}_c\M{Z}\Tra) = \tr(\M{N}\M{L}_c\M{N}\Tra) = 0$. Since $\M{N} \in \mathcal{F}^\perp$, we conclude that $\M{N} = \M{0}$.


\end{proof}

\subsection{\Prop{exact_recovery}}


\begin{proof}
Note that the objective function in \Eqn{bmc} is always nonnegative and that $\M{Z} = \M{X}$ attains this lower bound. Since \As{missingness} holds, by \Thm{existence} we can assert that $\M{Z} = \M{X}$ is the unique global minimizer of problem \Eqn{bmc}.
\end{proof}

\subsection{\Thm{existence}}

It is easier to work with vectorized quantities. Let $\V{x} = \vec(\M{X})$, namely $\V{x}$ is the vector obtained by stacking the columns of $\M{X}$ on top of each other. We observe that $\tr(\M{Z}\Tra\M{L}_r\M{Z}) =  \V{Z}\Tra ( \M{I}\Kron\M{L}_r) \V{Z}$ and $\tr(\M{Z}\M{L}_c\M{Z}\Tra) = \V{Z}\Tra (\M{L}_c \Kron \M{I}) \V{Z}$ where $\Kron$ denotes the Kronecker product. Then the objective in \Eqn{bmc} can be written as
\begin{eqnarray}
\label{eq:bmc_vec}
\min_{\V{z}}\>\frac{1}{2} \lVert \M{P}_\Omega \V{X} - \M{P}_\Omega \V{Z} \rVert_{2}^2 + \frac{\gamma_r}{2} \V{Z}\Tra ( \M{I}\Kron\M{L}_r) \V{Z} +  \frac{\gamma_c}{2} \V{Z}\Tra (\M{L}_c \Kron \M{I}) \V{Z}. 
\end{eqnarray}

We first establish that \Eqn{bmc_vec} always has a solution. Recall the edge-incidence matrix of the row graph $\M{\Phi}_r \in \Real^{\lvert E_r \rvert \times n}$ encodes its connectivity and is defined as
\begin{eqnarray}
\ME{\phi}{r,li} = \begin{cases}
\sqrt{w_l} & \text{If vertex $i$ is the head of edge $l$,} \\
-\sqrt{w_l} & \text{If vertex $i$ is the tail of edge $l$,} \\
0 & \text{otherwise.}
\end{cases}
\end{eqnarray}
The column edge-incidence matrix $\M{\Phi}_c \in \Real^{\lvert E_c \rvert \times p}$ is defined similarly. Recall that the Laplacian matrix of a graph can be written in terms of the edge-incidence matrix of the graph. Thus, Laplacian row matrix can be expressed as $\M{L}_r = \M{\Phi}_r\Tra \M{\Phi}_r$, and the Laplacian column matrix can be expressed as $\M{L}_c = \M{\Phi}_c\Tra \M{\Phi}_c$.
With these facts in hand, we can rewrite \Eqn{bmc_vec} as the following least squares problem.
\begin{eqnarray}
\label{eq:bmc_ls}
\min_{\V{z}}\>\frac{1}{2} \left \lVert \Vtilde{x} - \M{C}\V{z} \right \rVert_2^2,
\end{eqnarray}
where 
\begin{eqnarray*}
\Vtilde{x} & = & \begin{pmatrix}
\M{P}_\Omega \\ \M{0} \\ \M{0} \end{pmatrix}\V{x} \quad\quad \text{and} \quad\quad 
\M{C}  \amp = \amp \begin{pmatrix} \M{P}_\Omega \\
\sqrt{\gamma_r} \M{I}\Kron\M{\Phi}_r \\
\sqrt{\gamma_c} \M{\Phi}_c \Kron \M{I}
\end{pmatrix}.
\end{eqnarray*}
Recall that a least squares problem always has a solution since its solution is a Euclidean projection onto a closed convex set, namely the column space of the design matrix.

Having established that \Eqn{bmc_vec} always has a solution, we next characterize its solutions. A vector $\V{z}$ is a solution to \Eqn{bmc_vec} if and only if it satisfies the linear system in~\Eqn{smooth_lin}. 
where $\M{S}$ is defined in~\Prop{dof}.
The linear system in \Eqn{smooth_lin} may not be invertible. Since $\M{L}_r$ and $\M{L}_c$ are positive semidefinite,~\citet[Theorem 13.12]{laub2005matrix} guarantees that $\M{I} \Kron \M{L}_r$ and $\M{L}_c \Kron \M{I}$ are also positive semidefinite; therefore, $\M{S}$ is the sum of three positive semidefinite matrices, each of which may be rank deficient. 

We now state the conditions under which the linear system \Eqn{smooth_lin} has a unique solution. 

\begin{lemma}
\label{lem:l_S_spd}
Assume that $\gamma_r, \gamma_c > 0$. Then $\M{S}$ is positive definite if and only if
\begin{eqnarray*}
\Ker(\M{P}_\Omega) \cap \Ker(\M{I} \Kron \M{L}_r) \cap \Ker(\M{L}_c \Kron \M{I}) & = & \{\V{0}\}.
\end{eqnarray*}
\end{lemma}
\begin{proof}
Consider the quadratic form $\V{v}\Tra\M{S}\V{v}$.
\begin{eqnarray*}
\V{v}\Tra\M{S}\V{v} & = & \V{v}\Tra \left [\M{P}_\Omega + \gamma_r (\M{I} \Kron \M{L}_r) + \gamma_c (\M{L}_c \Kron \M{I}) \right ]\V{v} \\
& = & \V{v}\Tra \M{P}_\Omega \V{v} + \gamma_r \V{v}\Tra(\M{I} \Kron \M{L}_r) \V{v} + \gamma_c \V{v}\Tra (\M{L}_c \Kron \M{I}) \V{v}.
\end{eqnarray*}
Since $\M{P}_\Omega, \M{I} \Kron \M{L}_r$, and $\M{L}_c \Kron \M{I}$ are positive semidefinite, $\V{v}\Tra\M{S}\V{v} \geq  0$, with equality if and only if
\begin{eqnarray*}
\V{v}\Tra \M{P}_\Omega \V{v} \amp = \amp 
 \V{v}\Tra(\M{I} \Kron \M{L}_r) \V{v}
 \amp = \amp 
 \V{v}\Tra (\M{L}_c \Kron \M{I}) \V{v}
 \amp = \amp 0.
\end{eqnarray*}
If $\Ker(\M{P}_\Omega) \cap \Ker(\M{I} \Kron \M{L}_r) \cap \Ker(\M{L}_c \Kron \M{I}) = \{\V{0}\}$, then $\V{v}\Tra \M{S} \V{v} = 0$ implies that $\V{v} = \V{0}$ and $\M{S}$ has full rank. If $\Ker(\M{P}_\Omega) \cap \Ker(\M{I} \Kron \M{L}_r) \cap \Ker(\M{L}_c \Kron \M{I}) \not = \{\V{0}\}$, then there is a $\V{v} \not = \V{0}$ such that $\V{v}\Tra\M{S}\V{v} = 0$, and therefore $\M{S}$ is not invertible.
\end{proof}

We are now ready to prove \Thm{existence}.

\begin{proof}


Recall that $\Ker(\M{L}_r) = \text{Span}\{ \V{\chi}_{A_1}, \ldots, \V{\chi}_{A_R}\}$ and $\Ker(\M{L}_c) = \text{Span}\{ \V{\chi}_{B_1}, \ldots, \V{\chi}_{B_C}\}$. From~\cite[Theorem 13.16]{laub2005matrix}, it follows that 
\begin{eqnarray*}
\Ker(\M{I} \Kron \M{L}_r) \cap \Ker(\M{L}_c \Kron \M{I}) & = & \text{Span}\{ \V{\chi}_{B_c} \Kron \V{\chi}_{A_r} : r = 1, \ldots, R, c = 1, \ldots, C \}.
\end{eqnarray*}
\Lem{l_S_spd} implies that $\M{S}$ is positive definite if and only if $\M{P}_\Omega[\V{\chi}_{B_c} \Kron \V{\chi}_{A_r}] \not = \V{0}$ for all $r = 1, \ldots, R$ and $c = 1, \ldots, C$. Because
$\vec (\V{\chi}_{A_r}\V{\chi}_{B_c}\Tra) = \V{\chi}_{B_c} \Kron \V{\chi}_{A_r}$, this latter condition is equivalent to \As{missingness}.


If \As{missingness} does not hold, the least squares problem \Eqn{bmc_ls} will have infinitely many solutions since $\M{S}$ is positive semidefinite and the system is consistent.
\end{proof}

\subsection{\Prop{constrained_solution}}

\begin{proof}
Note that \Eqn{constrained} is equivalent to the following unconstrained problem.
\begin{eqnarray}
\label{eq:unconstrained}
\underset{\mu_{rc}}{\min}\; \frac{1}{2} \left \lVert \M{P}_\Omega \left (\sum_{r=1}^R\sum_{c=1}^C\mu_{rc}\V{\chi}_{A_r}\V{\chi}_{B_c}\Tra \right)  - \M{P}_\Omega (\V{x})\right \rVert_2^2.
\end{eqnarray}

Since the product set $\{A_1, \ldots, A_R\} \times \{B_1, \ldots, B_C\}$ is a partition of the index set $\{1, \ldots, n\} \times \{1, \ldots, p\}$, it follows that
$[\V{\chi}_{A_r} \V{\chi}_{B_c} \Tra ]_{ij} = 1$ if $i \in A_r$ and $j \in B_c$ and is 0 otherwise.
Using this fact the optimization problem in \Eqn{unconstrained} can be written as
\begin{eqnarray}
\underset{\mu_{rc}}{\min}\; \frac{1}{2} \sum_{r=1}^R\sum_{c=1}^C \sum_{(i,j) \in \Omega_{rc}} (\mu_{rc} - \VE{x}{ij} )^2.
\end{eqnarray}
The above problem separates over $\mu_{rc}$, and we can solve each problem individually.
\begin{eqnarray}
\underset{\mu_{rc}}{\min}\; \frac{1}{2} \sum_{(i,j) \in \Omega_{rc}} (\mu_{rc} - \VE{x}{ij} )^2.
\end{eqnarray}
Since \As{missingness} holds, we observe at least one entry for every partition. Therefore, each univariate optimization requires minimizing a strongly convex function. Consequently, the problem \Eqn{constrained} has a unique global minimizer. Elementary calculus shows that the solution is given by \Eqn{solution}.
\end{proof}

\subsection{\Thm{limiting_solution}}

\begin{proof}
Let $\Vn{\gamma}{m}$ be an arbitrary sequence of regularization parameters such that $\underline{\gamma}^{(m)} \rightarrow \infty$, and let $\Vn{z}{m}$ denote the solution to the biclustered matrix completion problem with $\Vn{\gamma}{m}$. Our proof proceeds in three steps.

\begin{description}
\item [Step 1: The sequence $\{\Vn{z}{m}\}$ has at least one limit point.]
We show that the sequence $\Vn{z}{m}$ is bounded and therefore resides in a compact set. First note that all but finitely many $\underline{\gamma}^{(m)} \geq 1$ since  $\underline{\gamma}^{(m)} \rightarrow \infty$. Next observe the following bound on the norm of $\Vn{z}{m}$,
$\lVert \Vn{z}{m} \rVert_2  \leq \lVert \M{S}\Inv \rVert_2 \lVert \M{P}_\Omega \V{x} \rVert_2$.
We just need to bound the operator norm of $\M{S}\Inv$. Note that $\M{S}  \succeq  \M{P}_\Omega + \M{L}_c\Kron\M{I} + \M{I}\Kron\M{L}_r$
for all $\V{\gamma}$ such that $\underline{\gamma} \geq 1$. The smallest eigenvalue of the matrix on the right, call it $\nu$, is strictly positive, therefore, $\lVert \M{S}\Inv \rVert_2 \leq \nu\Inv$ for all $m$ sufficiently large. Therefore all but finitely many $\Vn{z}{m}$ are within a Euclidean ball of radius $\nu\Inv\lVert \M{P}_\Omega \V{x} \rVert_2$. Consequently, $\Vn{z}{m}$ has at least one limit point.
\item [Step 2: Limit points of $\{\Vn{z}{m}\}$ are global minimizers of problem \Eqn{constrained}.]
We follow the argument in Theorem 17.1 of \cite{Nocedal2006}. Let $\V{z}^\star$ denote the vectorization of the unique solution $\M{Z}^\star$ to \Eqn{constrained}, then
\begin{eqnarray}
\label{eq:ineqA}
\frac{1}{2} \lVert \M{P}_\Omega \Vn{z}{m} - \M{P}_\Omega \V{x} \rVert_2^2 + \frac{\underline{\gamma}^{(m)}}{2} \psi(\Vn{z}{m}) & \leq &
\frac{1}{2} \lVert \M{P}_\Omega \V{z}^\star - \M{P}_\Omega \V{x} \rVert_2^2,
\end{eqnarray}
where
\begin{eqnarray*}
 \psi(\V{z}) & = & \V{z}\Tra(\M{I} \Kron \M{L}_r)\V{z} + \V{z}\Tra(\M{L}_c \Kron \M{I})\V{z}.
\end{eqnarray*}
Rearranging the inequality in \Eqn{ineqA} gives us
\begin{eqnarray}
\label{eq:ineqB}
\psi(\Vn{z}{m}) 
& \leq & \frac{1}{\underline{\gamma}^{(m)}}\left [\lVert \M{P}_\Omega \V{z}^\star - \M{P}_\Omega \V{x} \rVert_2^2 - \lVert \M{P}_\Omega \Vn{z}{m}- \M{P}_\Omega \V{x} \rVert_2^2\right ].
\end{eqnarray}
Let $\Vtilde{z}$ be a limit point of $\{\Vn{z}{m}\}$. Therefore, there is a subsequence $\mathcal{M}$ such that
$\underset{m \in \mathcal{M}}{\lim}\; \Vn{z}{m}  =  \Vtilde{z}$.
Taking limits of both sides of \Eqn{ineqB} along the subsequence $\mathcal{M}$ gives us
\begin{eqnarray}
\label{eq:ineqC}
\psi(\Vtilde{z}) & \leq & 
\underset{m \in \mathcal{M}}{\lim}\;  \frac{1}{\underline{\gamma}^{(m)}}\left [\lVert \M{P}_\Omega \V{z}^\star - \M{P}_\Omega \V{x} \rVert_2^2 - \lVert \M{P}_\Omega \Vn{z}{m}- \M{P}_\Omega \V{x} \rVert_2^2\right ] \amp = \amp 0.
\end{eqnarray}
The limit on the right hand side of  \Eqn{ineqC} is zero because $\underline{\gamma}^{(m)} \rightarrow \infty$ and the sequence $\{\Vn{z}{m}\}$ is bounded.
Since $\M{I} \Kron \M{L}_r$ and $\M{L}_c \Kron \M{I}$ are positive semidefinite, it follows that $\psi(\Vtilde{z}) \geq 0$ and consequently that $\Vtilde{z}$ is feasible for \Eqn{constrained}.

Finally, we argue that $\Vtilde{z}$ is not only feasible but also optimal. Since $\psi(\V{z}) \geq 0$ for all $\V{z}$, we have from the inequality in \Eqn{ineqA} that
\begin{eqnarray}
\label{eq:ineqD}
\frac{1}{2} \lVert \M{P}_\Omega \Vn{z}{m} - \M{P}_\Omega \V{x} \rVert_2^2 & \leq &
\frac{1}{2} \lVert \M{P}_\Omega \V{z}^\star - \M{P}_\Omega \V{x} \rVert_2^2,
\end{eqnarray}
Taking limits of both sides of \Eqn{ineqD} along the subsequence $\mathcal{M}$ gives us
\begin{eqnarray}
\frac{1}{2} \lVert \M{P}_\Omega \Vtilde{z} - \M{P}_\Omega \V{x} \rVert_2^2 & \leq &
\frac{1}{2} \lVert \M{P}_\Omega \V{z}^\star - \M{P}_\Omega \V{x} \rVert_2^2,
\end{eqnarray}
which establishes that $\Vtilde{z}$ is the global minimizer of \Eqn{constrained}. 
 
\item [Step 3: The sequence $\{\Vn{z}{m}\}$ converges to the global minimizer of problem \Eqn{constrained}.] We now have all the facts needed to prove the desired result. The sequence $\Vn{z}{m}$ has at least one limit point, and every limit point is a global solution to \Eqn{constrained}. There is, however, exactly one global solution to \Eqn{constrained} which means that $\Vn{z}{m}$ has exactly one limit point. Therefore, $\Vn{z}{m}$ converges to its one limit point, which is the global minimizer to \Eqn{constrained}. The global minimizer is given by \Eqn{global} and the proof is complete.
\end{description}
\end{proof}

\subsection{\Prop{dof}}

\begin{proof}
Let $\V{x} = \V{\mu} + \V{\varepsilon}$ where $\varepsilon_i$ are uncorrelated random variables with mean 0 and variance $\sigma^2$. Suppose, we have an estimator of $\V{\mu}$ that is a linear mapping of the observed response $\V{x}$, namely $\Vhat{y} = \M{A}\V{x}$ for some matrix $\M{A}$.
Consider the covariance between the $i$th elements of $\Vhat{x}$ and $\V{x}$.
\begin{eqnarray*}
\text{Cov}(\hat{x}_{i},\VE{x}{i}) & = & \mathbb{E}(\M{A}_{i\cdot}\V{\varepsilon}\varepsilon_i) \amp = \amp  \mathbb{E} \left (\sum_{j=1}^n \ME{a}{ij}\varepsilon_j\varepsilon_i \right ) \amp = \amp \ME{a}{ii} \sigma^2.
\end{eqnarray*}
Summing over these covariances and dividing by $\sigma^2$ verifies that the degrees of freedom of $\Vhat{y}$ is given by $\tr(\M{A})$.

Since $\V{z} = \M{S}\Inv\M{P}_\Omega\V{x}$ according to \Eqn{smooth_lin}, it follows that the degrees of freedom of $\V{z}$ is given by $\tr(\M{S}\Inv)$.
\end{proof}

\subsection{\Prop{dof_monotonicity}}

\begin{proof}


Note that if $\mathcal{G}_r$ has $R$ connected components and $\mathcal{G}_c$ has $C$ connected components, then $\M{Z}^\star$ has $RC$ degrees of freedom. 	

(i) Define $\M{A} = \M{I} +  \gamma_r (\M{I} \Kron \M{L}_r) + \gamma_c(\M{L}_c \Kron \M{I})$. Note that $\M{A} \succeq \M{S}$ and under \As{missingness} both $\M{A}$ and $\M{S}$ are positive definite; then~\cite[Corollary 7.7.4 (a)]{horn2013matrix} implies $\M{S}\Inv \succeq \M{A} \Inv$. 
The vectors $\V{\chi}_{B_c}\Kron\V{\chi}_{A_r}$ for $r = 1, \ldots, R, c = 1, \ldots, C$ are all eigenvectors of the matrix $\M{A}$ with eigenvalue 1. Therefore, $\V{\chi}_{B_c}\Kron\V{\chi}_{A_r}$ is an eigenvector of $\M{A}\Inv$ with eigenvalue 1. Since $\M{A}\Inv$ is positive definite all its eigenvalues are positive. Consequently, $\tr(\M{A}\Inv) \geq RC$ and therefore,~\cite[Corollary 7.7.4 (c)]{horn2013matrix} implies $\tr(\M{S}\Inv) \geq RC$.

(ii) Note that $\M{S}(\gamma'_r,\gamma'_c) \succeq \M{S}(\gamma_r,\gamma_c)$ if $\gamma'_r \geq \gamma_r, \gamma'_c \geq \gamma_c$. Therefore, ~\cite[Corollary 7.7.4 (c)]{horn2013matrix} implies $\tr(\M{S}\Inv(\gamma'_r,\gamma'_c)) \leq \tr(\M{S}\Inv(\gamma_r,\gamma_c))$. In words, as expected the degrees of freedom decreases as the amount of regularization increases.

(iii) Let $\M{M} = \M{S}\Inv$. We just need to show that 
$\lim_{k \rightarrow \infty} \text{Cov}(\MnE{Z}{k}{l}, \ME{X}{l}) =  \text{Cov}(\ME{Z}{l}^\star, \ME{X}{l})$, namely that   we can exchange the limit operation and the expectation operation. We will invoke the dominated convergence theorem to ensure that the two operations can be exchanged. For notational convenience, let $\VnE{q}{k}{l} \equiv ([\Mn{M}{k}(\V{x} - \V{\mu})]_l)(x_l - \mu_l)$ denote $\text{Cov}(\MnE{Z}{k}{l}, \ME{X}{l})$. Then, $\VnE{q}{k}{l}$ converges almost surely to $(\VE{z}{l}^{\star} - \mu_l)(x_l-\mu_l)$ since $\Vn{Z}{l} = \Mn{M}{l}\V{x}$ converges almost surely to $\V{Z}^\star$ by \Thm{limiting_solution}. We next identify a nonnegative random variable with finite expectation that bounds $\lvert \VnE{q}{k}{l} \rvert$ for all but finitely many $k$.
\begin{eqnarray*}
\lvert \VnE{q}{k}{l} \rvert & = &
\lvert [\Mn{M}{k}\V{\varepsilon}]_{l} \VE{\varepsilon}{l} \rvert
\amp = \amp \lvert \langle \Mn{M}{k}_{l\cdot}, \V{\varepsilon} \rangle \VE{\varepsilon}{l} \rvert
\amp \leq \amp \lVert \Mn{M}{k}_{l\cdot} \rVert_1 \lVert \V{\varepsilon} \rVert_\infty \lvert  \VE{\varepsilon}{l} \rvert 
\amp \leq \amp \lVert \Mn{M}{k} \rVert_\infty \lVert \V{\varepsilon} \rVert_\infty^2,
\end{eqnarray*}
where we have used H\"older's inequality. Note that   $\lVert \M{M} \rVert_\infty \leq \sqrt{np} \lVert \M{M} \rVert_2$ for all $\M{M}$, by the equivalence between the infinity and operator matrix norms. Thus,
\begin{eqnarray*}
\lvert \VnE{q}{k}{l} \rvert & \leq & \sqrt{np} \lVert \Mn{M}{k} \rVert_2 \lVert \V{\varepsilon} \rVert_\infty^2.
\end{eqnarray*}
Without loss of generality, we can take $\gamma^{(k)}_r, \gamma^{(k)}_c \geq 1$ since $\underline{\gamma}^{(k)} \rightarrow \infty$. Therefore, 
\begin{eqnarray*}
\Mn{M}{k} & \preceq & \left [\M{P}_\Omega + \M{I} \Kron \M{L}_r + \M{L}_c \Kron \M{I} \right]\Inv.
\end{eqnarray*}
Let $\nu$ denote the largest eigenvalue of the matrix on the right hand side. Then 
$\lVert \Mn{M}{k} \rVert_2 \leq \nu$ for all $k$. Applying this bound gives us a final bound on $\lvert \VnE{q}{k}{l} \rvert$ that is independent of $k$ and $l$, i.e.\@, 
$\lvert \VnE{q}{k}{l} \rvert  \leq  \sqrt{np}  \nu \lVert \V{\varepsilon} \rVert_\infty^2$.

Since $\mathbb{E} \lVert \V{\varepsilon} \rVert_\infty^2 < \infty$, we have that $\lim_{k \rightarrow \infty} \text{Cov}(\MnE{Z}{k}{l}, \ME{X}{l}) =  \text{Cov}(\ME{Z}{l}^\star, \ME{X}{l})$ for all $l$ by the dominated convergence theorem. 
\begin{eqnarray*}
\underset{k \rightarrow \infty}{\lim}\>\tr(\M{S}(\Vn{\gamma}{k})\Inv)  & = & 
\underset{k \rightarrow \infty}{\lim}\> \frac{1}{\sigma}\sum_{l=1}^{np} \text{Cov}(\MnE{Z}{k}{l}, \ME{X}{l}) \amp = \amp
\frac{1}{\sigma} \sum_{l=1}^{np} \text{Cov}(\ME{Z}{l}^\star, \ME{X}{l})
\amp = \amp
RC.
\end{eqnarray*}
\end{proof}

\bigskip
\begin{center}
{\large\bf SUPPLEMENTAL MATERIALS}
\end{center}

\begin{description}

\item[Algorithm Derivations:] The Supplementary Materials includes additional details on derivations needed to implement the Quasi-Newton method as well as additional simulation experiments. (\url{https://github.com/echi/IMS/blob/master/BMC_Supplement_JCGS.pdf})

\item[Code:] Matlab code implementing IMS and scripts for regenerating the numerical results are available at \href{https://github.com/echi/IMS}{https://github.com/echi/IMS}.

\end{description}

\begin{center}
{\large\bf ACKNOWLEDGMENTS}
\end{center}
The authors acknowledge Salman Asif and Chris Harshaw for their help on a prior project from which this current work arose.

\bibliographystyle{asa}
\bibliography{bmc}

\end{document}